\documentclass[aps,prb,twocolumn,10pt,superscriptaddress]{revtex4-2}
\usepackage{header}
\usepackage{adjustbox}
\usepackage{physics}
\newtheorem{theorem}{Theorem}[section]
\newtheorem{lemma}[theorem]{Lemma}
\newtheorem{protocol}[theorem]{Protocol}
\newtheorem{corollary}[theorem]{Corollary}

\newcommand{\ketb}[1]{\ket{\mathbf{#1}}}

\begin{document}
\title{Fusion for High-Dimensional Linear Optical Quantum Computing with Improved Success Probability}
\author{G\"ozde \"Ust\"un}
\affiliation{School of Electrical Engineering and Telecommunications, UNSW Sydney, Sydney,
NSW 2052, Australia}
\affiliation{ARC Centre of Excellence for Quantum Computation and Communication Technology, Melbourne, VIC, Australia}
\author{Eleanor G. Rieffel}
\affiliation{QuAIL, NASA Ames Research Center, Moffett Field, CA 94035, USA}
\author{Simon J. Devitt}
\affiliation{Centre for Quantum Software and Information, University of Technology Sydney, Sydney, New South Wales 2007, Australia}
\affiliation{InstituteQ, Aalto University, 02150 Espoo, Finland}
\author{Jason Saied}
\affiliation{QuAIL, NASA Ames Research Center, Moffett Field, CA 94035, USA}
\begin{abstract}
Type-II fusion is a probabilistic entangling measurement that 
is essential to measurement-based linear optical quantum computing and can be used for quantum teleportation more broadly. 
However, it remains under-explored for high-dimensional qudits. 
Our main result gives a Type-II fusion protocol with proven success probability approximately $2/d^2$ for qudits of arbitrary dimension $d$. 
This generalizes a previous method which only applied to even-dimensional qudits. 
We believe this protocol to be the most efficient known protocol for Type-II fusion, with the $d=5$ case beating the previous record by a factor of approximately $723$. 
We discuss the construction of the required $(d-2)$-qudit ancillary state using a silicon spin qudit ancilla coupled to a microwave cavity through time-bin multiplexing. 
We then introduce a general framework of extra-dimensional corrections, a natural technique in linear optics that can be used to non-deterministically correct non-maximally-entangled projections into Bell measurements. 
We use this method to analyze and improve several different circuits for high-dimensional Type-II fusion and compare their benefits and drawbacks. 
\end{abstract}
\maketitle
\section{Introduction}
Photonic hardware is a cutting-edge and highly promising platform for building a scalable quantum computer~\cite{rudolph2017optimistic,psiQ}. It has the potential to be far more scalable than solid-state hardware, as millions of photons can be generated from a single device. However, entangling photons is difficult due to their inability—or extremely weak ability—to interact with one another \cite{Kok_2007, Ustun2025}.

Despite this, individual photons can naturally spread across multiple modes, enabling the creation of useful high-dimensional states deterministically~\cite{rudolph2021everywhereundetectabledistributedquantum,Ustun2025, reck1994experimental, sophia_2024}. 
Further, for a $d$-dimensional linear optical qudit, arbitrary single-qudit unitaries may be performed deterministically using only linear optical elements. 
These ingredients may be combined with a single entangling operation to 
generate arbitrary multi-qudit states and perform universal high-dimensional quantum computing \cite{Paesani, Lib_2024}. 
Since a single photon may encode a high-dimensional qudit (equivalent to a large number of qubits), the higher-dimensional setting is attractive for quantum networking and computing \cite{Cabrejo_Ponce_2023,Zahidy2024,Lib_2025,Lib_2024,sophia_2024, shalaev2025photonicnetworkingquantummemories, Yamazaki}. 
In principle, photon loss makes such encodings less desirable; however, introducing redundancy as in quantum error correction may circumvent this issue. 
For example, the $[[5,1,3]]_{\mathbb{Z}_d}$ modular-qudit code \cite{chau1997five} enables one to use $5$ physical ($d$-dimensional) qudits to encode one $d$-dimensional logical qudit and can tolerate $2$ erasures (equivalent to photon losses) regardless of dimension. In other words, replacing the physical systems with qudits allows for the encoding of more information, with the same robustness to error, without increasing the number of physical systems required. 
A similar analysis holds for the qudit surface code \cite{bullock2007qudit}. 

Qudit graph states can potentially serve as a resource for high-dimensional analogues of fusion-based quantum computing (FBQC), 
a universal model for quantum computation that has been well studied in the qubit case~\cite{Bartolucci2023,bombin2021interleaving, litinski2022active, bombin2023fault}. FBQC involves taking small photonic resource states and performing entangling measurements (called \emph{Type-II fusions}) to carry out a version of fault-tolerant measurement-based quantum computation. 
The prototypical Type-II fusion is a non-deterministic Bell measurement \cite{Terry_2005}; this has been generalized to include measurements projecting onto other maximally entangled states \cite{bartolucci2021creationentangledphotonicstates}. 
In this work, we consider further generalizations that can utilize projections onto non-maximally-entangled states (see Section~\ref{sec:USD}). 

Extensive research has been conducted on the effectiveness of linear optical measurements that project onto the four standard dual-rail Bell basis states. Early studies \cite{Calsamiglia_2001} demonstrated that Bell state measurements (BSMs) performed with only vacuum ancillae can achieve, at most, 50\% efficiency in distinguishing all four Bell states. However, substantial improvements \cite{Grice, Ewert_2014} have been achieved by incorporating ancillary photonic states, allowing the efficiency to exceed the 50\% threshold. 
For example, the protocol of \cite{Ewert_2014} can boost the success probability to 62.5\% using two ancillary single photons. (Or 75\% with four ancillary single photons.) 
More advanced approaches to Type-II fusion were later introduced in \cite{bartolucci2021creationentangledphotonicstates}, including a boosted Type-II fusion scheme that is not a conventional BSM, but achieves a higher success probability of 66.6\% using two ancillary single photons. 
In principle, sufficiently entangled ancillae allow one to boost the success probability arbitrarily close to $1$ \cite{Ewert_2014}. 

For linear optical qudits, however, only a few studies have explored high-dimensional Type-II fusion, through the perspective of quantum teleportation. %
These existing approaches either require some form of non-linearity~\cite{goyal2014qudit,Goel2024}, or work only for even-dimensional qudits \cite{bharos2024efficienthighdimensionalentangledstate, bharoshigh}, or have success probability rapidly decaying with the qudit dimension \cite{Luo}. 
(We note that, despite claims in subsequent literature, the work of \cite{Luo} does not claim a success probability scaling like $1/d^2$; rather, it decays exponentially: see Section~\ref{sec:luo}.) 
The absence of an efficient linear optical Type-II fusion gate for odd-dimensional qudits is a serious gap in the literature, as much work on qudits requires the dimension to be prime or a prime power. 
When the dimension is not a prime power, many familiar features from the qubit setting no longer apply: for example, the analogue of the Clifford group is not a 2-design \cite{heinrich2021stabiliser}, and it is unknown whether there exist complete sets of mutually unbiased bases \cite{mcnulty2024mutually}. 
This leads to obstacles in quantum state tomography and other areas \cite{rambach2021robust, lima2011experimental, bent2015experimental, peres2023pauli}. 

In this work, we extend the protocol of \cite{bharos2024efficienthighdimensionalentangledstate} to construct a Type-II fusion gate for qudits of \emph{arbitrary} dimension $d$. We prove that in the new case in which $d$ is odd, the success probability is $2/d(d+1)$ (see Section~\ref{sec:bharos odd}). 
This is the \textbf{highest success probability recorded to date} and allows for fusion in the important prime-dimensional case. 
For illustration, in the case $d=5$ our method has a success probability over $723$ times higher than that of \cite{Luo}, which we believe to be the previous record for linear-optical Type-II fusion. 
Note the scaling we find in the odd-dimensional case is comparable with the even-dimensional case studied by \cite{bharos2024efficienthighdimensionalentangledstate, bharoshigh}, which has success probability $2/d^2$. 
We also provide an example demonstrating how to construct the required ancilla (equivalent to a lower-dimensional GHZ state of $d-2$ or $d-1$ photons in the even and odd cases respectively) using a realistic physical system. 

In subsequent sections, we consider other protocols for high-dimensional Type-II fusion, which avoid the need for multi-photon entangled ancillae but have lower success probability than the protocols of \cite{bharos2024efficienthighdimensionalentangledstate, bharoshigh} and our Section~\ref{sec:bharos odd}. 
This is motivated by the method of~\citet{Luo}, which uses multiple single-photon ancillae rather than one many-photon entangled ancilla, and further does not project onto a maximally entangled state: instead, the protocol requires certain non-deterministic, adaptive corrections (beyond simply phases). 
We establish a rigorous framework extending this method of \emph{extra-dimensional corrections}, in which ancillary vacuum modes are appended to a linear optical state, an interferometer is applied on the larger set of modes, and post-selection is performed to filter out terms nontrivially utilizing the extra modes. 
This is a relatively straightforward operation in linear optics, similar to appending ancillary modes for a standard (boosted) Type-II fusion or entangled state generation circuit. 
Extra-dimensional corrections are a concrete linear optical implementation of Procrustean distillation as in \cite{vidal1999entanglement}. 
This method can be used to correct a larger class of projections to Bell measurements and improve the overall success probability. 

We apply extra-dimensional corrections, improving the framework of \citet{Luo} for $d>3$, as seen in Table~\ref{table:luo} and Section~\ref{sec:scaling}. 
We numerically show that this extended protocol has significantly lower success probability than the Type-II fusion protocol we give in Section~\ref{sec:bharos odd} (which \emph{does not} require such corrections). 
However, the protocol of Section~\ref{sec:scaling} may still be worth considering for certain applications, since it requires only $d-2$ single-photon ancillae (rather than a multi-photon entangled state). 

We also consider a high-dimensional Type-II fusion protocol derived from a Bell state generation circuit of \cite{Paesani}. 
This method, discussed in Section~\ref{sec:paesani}, requires extra-dimensional corrections and a single-mode, multi-photon bunched state $\ket{r}$ as an ancilla. 
As shown in Table~\ref{table:summary}, this protocol has an intermediate success rate between the other cases. 
In particular, if one has access to multi-photon entanglement, the protocols of Sec.~\ref{sec:bharos odd} (extending \cite{bharos2024efficienthighdimensionalentangledstate, bharoshigh}) are optimal and extra-dimensional corrections are not needed. 
If one cannot prepare multi-photon entangled ancillae, but \emph{can} prepare single-mode bunched states, then the methods of Section~\ref{sec:paesani} are the best known. 
Finally, if one has access only to single-photon ancillae, the best known protocol is that of Section~\ref{sec:scaling}. 

In Section~\ref{sec:psiq-increased}, we also analyze a new Type-II fusion protocol for \emph{qubits} ($d=2$), generalizing a protocol of \cite{bartolucci2021creationentangledphotonicstates}. 
This method, which uses extra-dimensional corrections (with at most one extra dimension needed), is notable because, numerically, its success probability seems to continue increasing as more single photon ancillae are added. 
However, for small numbers of ancillary photons, it is less effective than other methods such as \cite{Ewert_2014}. 

The paper is organized as follows. We begin with Section~\ref{sec:preliminaries}, giving background material on linear optics, FBQC, Type-II fusion, and teleportation. 
We also discuss the paradigm of \emph{Fourier projection}, allowing us to view the Type-II fusion protocols of \cite{Luo, bharos2024efficienthighdimensionalentangledstate, bharoshigh} as instances of a single family using different choices of ancillary state. 
Next, in Section~\ref{sec:bharos extension}, we present our main results. We review the work of \cite{bharos2024efficienthighdimensionalentangledstate, bharoshigh} in Section~\ref{sec:bharos even} and, in Section~\ref{sec:bharos odd}, we extend their techniques to give the optimal known Type-II fusion gate for odd-dimensional qudits. 
This essentially involves embedding the odd-dimensional qudits into a larger even-dimensional system. 
In Section~\ref{sec:ancilla}, we then propose a physical system for constructing the ancilla required for our Type-II fusion gate and that of \cite{bharos2024efficienthighdimensionalentangledstate, bharoshigh}. 
Section~\ref{sec:implement} discusses the feasibility of near-term implementations of our protocol for small dimensions $d$. 
In Section~\ref{sec:USD}, we introduce the notion of extra-dimensional corrections. 
Section~\ref{sec:luo} reviews the work of \cite{Luo}, applies extra-dimensional corrections to increase the success probability, and clarifies a common but incorrect assumption in the literature: we show that both the protocol of \cite{Luo}, and its improvement using extra-dimensional corrections, have success probability far smaller than $1/d^2$ for general $d$. 
We subsequently consider alternative Type-II fusion gates in Sections~~\ref{sec:paesani} and \ref{sec:psiq-increased} and Appendices~\ref{app:paesani} and \ref{sec:boosted_high_dim}. 
The last is a ``boosted'' fusion generalizing the Type-II fusion of \citet{Luo} and the qubit boosting protocol of \citet{Ewert_2014}. 
We then summarize the work in Section~\ref{sec:conclusion}.
A summary table giving the success probabilities of many of the Type-II fusion protocols we consider here is given in Table~\ref{table:summary}. 

\begin{table}[htb]
\begin{tabular}{|c|c|c|c|}
\hline
 & $d=3$ & $d=4$ & $d=5$ \\ \hline
\cite{bharos2024efficienthighdimensionalentangledstate, bharoshigh} & N/A & \textbf{0.125}; 2 & N/A \\ \hline
Sec.~\ref{sec:bharos odd} & \textbf{0.166}; 2 & \textbf{0.125}; 2 & \textbf{0.067}; 4 \\ \hline
\cite{Luo} & 0.111; 1 & $9.8\times 10^{-4}$; 2 & $9.2\times 10^{-5}$; 3 \\ \hline
Sec.~\ref{sec:scaling} & 0.111; 1 & 0.017; 2 & 0.003; 3 \\ \hline
Sec.~\ref{sec:paesani} & 0.116; 1 & 0.020; 2 & 0.004; 3 \\ \hline
Sec.~\ref{sec:paesani} & 0.140; 4 & 0.056; 5 & 0.018; 5 \\ \hline
\end{tabular}
\caption{\small We consider the main Type-II fusion protocols discussed in the present work. For qudit dimension $3\leq d\leq 5$, we give the largest success probability found using that protocol, followed by the number of ancillary photons required. 
We give two variants of the results in Sec.~\ref{sec:paesani}: one minimizing the number of ancillary photons and the other maximizing the success probability. 
The first two rows use a multi-photon entangled ancilla; the next two use multiple single-photon ancillae; the final two use a one-mode, multi-photon bunched state $\ket{r}$. 
The success probabilities for \cite{bharos2024efficienthighdimensionalentangledstate, bharoshigh} and our Sec.~\ref{sec:bharos odd} are given in boldface, as they are the largest known in the appropriate dimensions. These two protocols are also unique because they do \emph{not} require the extra-dimensional corrections discussed in Sec.~\ref{sec:USD}. 
\label{table:summary}
}
\end{table}

\section{Preliminaries}\label{sec:preliminaries}
Note: If the reader is familiar with the fundamental concepts of fusion-based quantum computation for qubits, such as Type-II fusion and boosted variants, they may proceed directly to Section~\ref{sec:fourier}. 

\subsection{Linear Optics Notation}\label{sec:notation}
We briefly set out notation used throughout the work. We use $\ket{n_0 n_1 \cdots n_{m-1}}$ to denote the $m$-mode Fock state with $n_0$ photons in mode $0$, $n_1$ photons in mode $1$, etc. 
We will generally be concerned with \emph{linear optical qudits}, a generalization of dual-rail qubits. For qudits of dimension $d$, we define the linear optical qudit computational basis states by 
\begin{equation}
    \ketb{k} = |00\cdots 0 \overset{\overset{k}{\downarrow}}{1} 0 \cdots 0\rangle, 
\end{equation}
with a $1$ in the $k$th index, for $0\leq k\leq d-1$. 
If we are partitioning a linear optical state into multiple qudits, we will often use commas to separate the appropriate modes. For example, with qudit dimension $d=4$, the following are equivalent expressions for the same state with photons in modes $0$ and $5$: 
\begin{equation}
    \ketb{01} = \ket{10000100} = \ket{1000,0100} = \ketb{0,1}.
\end{equation}
We also introduce the qudit \emph{Pauli operators} in dimension $d$, 
\begin{align*}
    X &= \sum_{k=0}^{d-1} \ketb{k+1}\bra{\mathbf{k}}, 
    \\ Z &= \sum_{k=0}^{d-1} \omega^{k}\ketb{k}\bra{\mathbf{k}}, 
\end{align*}
where $\omega$ is a fixed primitive $d$th root of unity and $\ketb{k+1}$ is interpreted modulo $d$. 

We now consider the space of $n$ linear optical qudits, each of dimension $d$. 
The first $d$ modes correspond to one qudit, the next $d$ modes to the next qudit, and so on. 
Often, we consider a physical setting in which the qudits are identified with different spatial modes or \emph{ports}, and the modes within a qudit correspond to \emph{time bins}. 
For convenience, and to avoid overloading the word ``mode," we often refer to the $d$ modes within a single qudit as \emph{time bins} regardless of the physical context. (The results apply regardless of how the modes are physically implemented.) 
Especially given this context, it is often helpful to index the modes by pairs $(i,j)$, corresponding to mode $d*i+j$, the $j$th ``time bin" within the $i$th qudit. 

The relevant transformations are the \emph{linear optical unitaries} (LOUs), which arise from interferometers and linearly transform creation operators. 
An important observation is that arbitrary LOUs \emph{do not} preserve the space of linear optical qudits. 
Single-qudit unitaries may easily be implemented by LOUs that act only on the relevant $d$ modes, but multi-qudit entangling operations cannot be implemented using LOUs alone. 
This will lead to the various notions of fusion discussed below, in which LOUs and post-selective measurements are used to non-deterministically implement entangling operations. 

We will often use the LOU corresponding to the discrete Fourier transform, 
\begin{equation}\label{eq:fourier}
    F_n = \begin{pmatrix}
        \omega_n^{ij}
    \end{pmatrix}_{0\leq i,j\leq n-1},
\end{equation}
where $\omega_n$ is a primitive $n$th root of unity. (When $n=d$, we generally write $\omega=\omega_n$.) 
This unitary may operate on any subset of $n$ modes.

\subsection{Fusion Based Quantum Computation for Qubits}
The design of a quantum computer hinges on the careful selection of fundamental physical operations and the way they are combined, shaping both the hardware and overall architecture. This becomes particularly important in fault-tolerant quantum computing, in which a small set of operations is applied repeatedly throughout the process. 
The physical limitations and error mechanisms from the underlying physical primitives can greatly affect the efficiency and performance of the logical-level quantum computation. 
Fusion-based quantum computing (FBQC) \cite{Bartolucci2023} is a model for universal fault-tolerant quantum computation with linear optics, designed to circumvent the lack of a deterministic entangling (unitary) gate in that setting. 
A non-deterministic \emph{gate} would necessitate repeating a circuit many times before obtaining a useful outcome \cite{knill2001scheme}; instead, FBQC utilizes Type-II fusion, a non-deterministic entangling \emph{measurement} that can be used in an adapted version of measurement-based quantum computation. 
(Although FBQC can in principle be applied to other physical systems, this motivation makes it most relevant to linear optics.)
FBQC requires two main ingredients: constant-size entangled states known as \emph{resource states}, 
and entangling measurements, referred to as \textit{Type-II fusions} (or just \emph{fusions}). In the simplest case, discussed below, fusion is a Bell state measurement, non-deterministically implemented using linear optics \cite{Terry_2005}. 

\begin{figure}[htb]
    \centering
    \includegraphics[width=\linewidth]{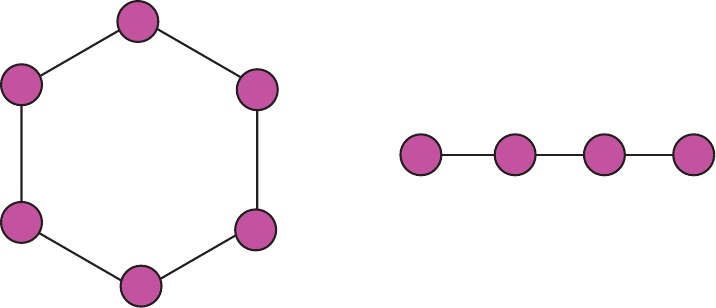}
    \caption{Examples of qubit graph states. Up to single-qubit Clifford operations, any stabilizer state is expressible as a graph state, so they are commonly used as entangled resource states. On the left is a graph state consisting of six entangled photonic qubits arranged in a hexagonal configuration, often referred to as a $6$-ring \cite{Bartolucci2023}. On the right is a graph state of four entangled photonic qubits, generally called a linear cluster state. 
}
    \label{fig:fusion}
\end{figure}
\begin{figure}[ht]
    \centering
    \includegraphics[width=\linewidth]{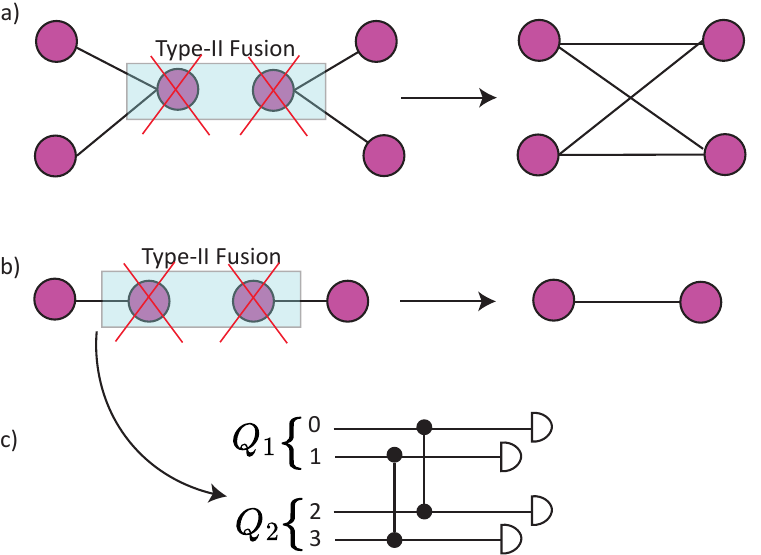}
    \caption{\small Type-II fusion:  
    a) We perform Type-II fusion between chosen qubits of two three-photon graph states. Type-II fusion destroys all the qubits that are involved - measured - in the process. The resulting graph state after the Type-II fusion is a graph state with nontrivial entanglement between the four surviving photons. b) For a simpler example, Type-II fusion can be used to join two linear cluster states into a single joint cluster. Here we obtain a Bell pair; if the initial linear clusters instead had lengths $n$ and $m$, the resulting linear cluster would have length $n+m-2$. 
    c) Standard Type-II fusion circuit: Qubits are represented using dual-rail encoding, where each qubit consists of one photon between two modes. Two 50:50 beamsplitters (Hadamard operations) are applied between modes 0 and 2, and modes 1 and 3, followed by measurements on all modes.}
    \label{fig:fusion-types}
\end{figure}

\subsection{Type-II Fusion}\label{sec:type-ii}
In Type-II fusion, we have two states $\ket{\alpha}, \ket{\beta}$, each involving many linear optical qudits in general. 
We fix a qudit dimension $d$. 
In order to create a larger entangled state, we wish to perform a Bell measurement involving two qudits, one from $\ket{\alpha}$ and one from $\ket{\beta}$. 
In particular, we want to project the two input qudits onto the ideal Bell state $\ket{B_0}= \frac{1}{\sqrt{d}}\sum_{k=0}^{d-1}\ketb{kk}$, destroying them in the process and creating entanglement between the remaining qudits of $\ket{\alpha}$ and $\ket{\beta}$. 
(See Fig.~\ref{fig:fusion-types}a and Fig.~\ref{fig:fusion-types}b.) 

Type-II fusion was introduced in the qubit setting by \cite{Terry_2005} to overcome the fact that Bell measurements cannot be deterministically implemented in linear optics. 
Instead, linear optics and photon number-resolving detection (PNRD) are used to implement a POVM on the two input qudits: some of the POVM elements correspond to a heralded \emph{fusion failure}, projecting onto an unentangled state; others lead to \emph{success}, projecting onto $(1\otimes P)\ket{B_0}$, where $P$ is a known single-qubit Pauli operator determined by the outcome. 
This Pauli factor may then be deterministically corrected by applying appropriate Pauli operators to the surviving qubits. 
For random input states, this Type-II fusion protocol has a success probability of $1/2$. 
A variant of the original Type-II fusion circuit (without the corrections) is given in Fig.~\ref{fig:fusion-types}c. 
Following a standard convention, the gates depicted are 50:50 beamsplitters with Hadamard transfer matrix, and the $D$ symbols at the right side of each mode indicate PNRD. 
By performing PNRD across all modes of the input qubits, Type-II fusion enables robust detection of photon number errors, making it particularly useful in practical implementations \cite{Terry_2005}.

Subsequent works increase the probability of success (still in the qubit case) by a method called \emph{boosting}, introducing ancillary photons and a more complicated POVM \cite{Grice, Ewert_2014}. 
In the method of \cite{Ewert_2014}, 
the key is to append an ancillary state and additional beamsplitters after the original fusion circuit but before PNRD. The ancilla is chosen to preserve appropriate parities, which prevents previously successful fusion outcomes from turning into failures, while also converting some failure outcomes into successful Bell state projections. 
The chosen ancilla in this boosting protocol is the state $\ket{C^2} = \frac{1}{\sqrt{2}} (\ket{20} - \ket{02})$, which can be deterministically obtained by sending two single photons through different arms of a 50:50 beamsplitter. 
The circuit diagram of this boosted fusion, which has a success rate of $5/8 = 0.625$, is shown in Fig.~\ref{fig:boosted_type-II}. 
Note that the circuit uses two extra modes, coupled with the modes originally occupied by the second qubit. 
One may carry out the same procedure for the other qubit as well, using four ancillary photons in four ancillary modes; this increases the success probability to $3/4=0.75$. The details of this process, and higher-probability variants, can be found in~\cite{Ewert_2014}. 

Another boosting protocol for Type-II fusion is introduced in~\cite{bartolucci2021creationentangledphotonicstates}, depicted in Fig.~\ref{fig:boosted_type-II-2}. This method uses a single photon ancilla and incorporates a three-dimensional Fourier transformation. Unlike previous boosting protocols, it also allows for projections onto maximally entangled states that are \emph{not} Bell states. 
In the right circumstances, the resulting projections may be corrected by applying appropriate unitaries to the unmeasured qubits, achieving the originally desired Bell measurement. 
The success probability of this boosting protocol is 
$7/12$. 
Like the method of \cite{Ewert_2014}, one may carry out this procedure for \emph{both} qubits to increase the success probability to $2/3$. 
More details regarding this protocol can be found in~\cite{bartolucci2021creationentangledphotonicstates}. 

Such reasoning may be carried into the qudit case as well. We wish to project onto $\ket{B_0}= \frac{1}{\sqrt{d}}\sum_{k=0}^{d-1}\ketb{kk}$; instead, we perform a POVM, potentially involving ancillary photons. Certain outcomes project the target qudits onto unentangled states, leading to fusion failure. 
Others may project onto maximally entangled states, allowing for unitary corrections generalizing \cite{bartolucci2021creationentangledphotonicstates}. 
Further, we will see in Section~\ref{sec:USD} that projections onto many entangled states that are \emph{not} maximally entangled may also be corrected, if we are willing to accept a probability of failure in this step as well. 

\begin{figure}
    \centering
    \includegraphics[width=0.6\linewidth]{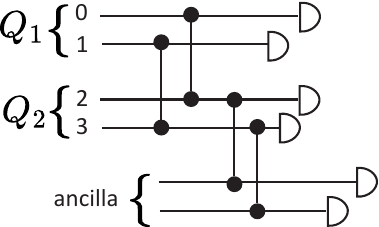}
    \caption{Boosted Type-II Fusion \cite{Ewert_2014}. 
    As in Fig.~\ref{fig:fusion-types}, we apply 50:50 (Hadamard) beamsplitters between corresponding modes of the input qubits. Before measuring, however, we also allow for interference with an ancillary state of the form $\frac{1}{\sqrt{2}}(\ket{20} - \ket{02})$. 
    This protocol has a success probability of $0.625$; if the other qubit undergoes a similar treatment, the success probability is increased to $0.75$. 
    }
    \label{fig:boosted_type-II}
\end{figure}

\begin{figure}
    \centering
    \includegraphics[width=0.6\linewidth]{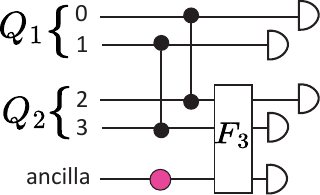}
    \caption{Boosted Type-II Fusion \cite{bartolucci2021creationentangledphotonicstates}. Similarly to Fig.~\ref{fig:boosted_type-II}, we allow for interference with an ancillary state between the standard fusion circuit and the measurement. This protocol uses only a single ancillary photon and a three-mode Fourier transform. 
    The success probability is $7/12$ as depicted and $2/3$ if the additional interference is implemented on both sides. 
}
    \label{fig:boosted_type-II-2}
\end{figure}
\subsection{Fusion, Teleportation, Bell State Generation}\label{sec:teleportation}
We now recall the well-understood connections between fusion, teleportation, and Bell state generation.

In an idealized teleportation setting, we have two parties Alice and Bob, each possessing half of a two-qudit Bell pair $\ket{B_0} = \sum_{k=0}^{d-1}\frac{1}{\sqrt{d}}\ketb{kk}$. Alice has some single-qudit state $\ket{\psi}$. By performing a Bell measurement on her two qudits ($\ket{\psi}$ and half of the Bell pair), she projects Bob's subsystem onto one of the states $P\ket{\psi}$, where $P$ is a single-qudit Pauli operator. 
Alice communicates the result of her Bell measurement to Bob, allowing him to determine the appropriate Pauli operator required to transform his subsystem into Alice's initial state $\ket{\psi}$. 

If the state $\ket{\psi}$ is a linear optical qudit, however, a perfect Bell measurement cannot be performed. 
Instead, Alice can perform a Type-II fusion between $\ket{\psi}$ and her half of the Bell pair $\ket{B_0}$. 
As discussed above, this process is non-deterministic, with the fusion itself and the resulting correction potentially leading to failure. 
If both succeed, however, Alice's state will be teleported onto Bob's qudit. 
The teleportation protocols of \cite{Luo, bharos2024efficienthighdimensionalentangledstate, bharoshigh} all implement non-deterministic Bell measurements (with known corrections), and therefore they may be viewed as protocols for Type-II fusion between two states $\ket{\alpha}, \ket{\beta}$. 
Teleportation is simply the special case in which $\ket{\alpha} = \ket{\psi}$ is a single-qudit state and $\ket{\beta}$ is a Bell pair. 
Not all teleportation protocols may be recast as Type-II fusion in this way, however; the work of \cite{zhang2019quantum}, for example, cannot obviously be viewed as a Bell measurement between two parties. 

We also recall that, since Type-II fusion is essentially a Bell state measurement (up to potentially complicated correction operations), it may be viewed as dual to Bell state generation. 
Then one may reverse fusion circuits to obtain Bell state generation protocols and vice versa, as long as one is careful about the choice of ancillary state and interpretation of the resulting measurements. 
Works such as \cite{bartolucci2021creationentangledphotonicstates, Paesani} exploit this duality, reversing fusion circuits to obtain Bell state generation protocols and vice versa. 
In Section~\ref{sec:paesani}, we exploit this duality to convert Bell state generation protocols of \cite{Paesani} into higher-dimensional fusion circuits and analyze their effectiveness. 
\subsection{Fourier Projection}\label{sec:fourier}
We now discuss a linear optical circuit we refer to as \emph{Fourier projection}, which is the main step in the qudit fusion protocols of \cite{Luo, bharos2024efficienthighdimensionalentangledstate, bharoshigh}. 
We consider a state of $d$ linear optical qudits, each of dimension $d$. 
This is a state of $d^2$ modes: we view the first $d$ modes as corresponding to one qudit, the next $d$ modes to the next qudit, and so on. 
Recall the pairwise indexing of the modes from Section~\ref{sec:notation}, in which $(i,j)$ corresponds to the $j$th mode of qudit $i$. 

\begin{protocol} The Fourier projection is given as follows. \label{proto:fourier}
\begin{enumerate}
    \item Input $d$ linear optical qudits of dimension $d$ as given above. 
    \item For each $0\leq j\leq d-1$, apply the Fourier interferometer \eqref{eq:fourier} on the set of modes $\{(i,j): 0\leq i\leq d-1\}$. In a physical setting in which the qudits correspond to spatial modes (ports) and the modes within a qudit correspond to time bins, this is equivalent to a Fourier interferometer acting only on the spatial modes and fixing the time bins. 
    \item Measure all $d^2$ modes using PNRD. 
    \item The Fourier projection is considered a \emph{success} if, for each $0\leq j\leq d-1$, exactly one photon is detected in the modes $\{(i,j): 0\leq i\leq d-1\}$. (In other words, we post-select for measurement outcomes in which each time bin receives exactly one photon.) Otherwise, the projection is a \emph{failure}. 
\end{enumerate}
\end{protocol}

Recalling that $F_2$ is simply the Hadamard matrix, in the $d=2$ case this is precisely the Type-II fusion gate of Fig.~\ref{fig:fusion-types}. 
The $d=3,4$ cases are depicted in Figs.~\ref{fig:luo-fusion} and \ref{fig:bharos-type-II} respectively (ignoring the choices of input states made there). 
In those figures, motivated by the $(i,j)$ indexing of the modes discussed above, the $d$ Fourier interferometers applied in Step 2 of Protocol~\ref{proto:fourier} are depicted as a single unitary $F_d\otimes I$. 


We index successful measurement patterns by $(q_0, q_1, \dots, q_{d-1})$, where $q_i$ is the index of the unique \emph{qudit} (or port) containing a photon in time bin $i$. 
We have the following, observed in \cite{Luo, bharos2024efficienthighdimensionalentangledstate, bharoshigh} and easily seen by direct computation: 
\begin{lemma}\label{lemma:fourier projection}
    For qudit input, successful Fourier projection with measurement pattern $(q_0, \dots, q_{d-1})$ projects onto the (unnormalized) state
    \begin{equation*}
        \ket{f_d^{(q)}} = \frac{1}{d^{d/2}}\sum_{\substack{k_0, \dots, k_{d-1}\\ \textnormal{distinct}}} \omega^{-\sum_{j=0}^{d-1} k_j q_j}\ketb{k_0, k_1, \dots, k_{d-1}}.
    \end{equation*}
\end{lemma}

In the works \cite{Luo, bharoshigh}, Fourier projection is used to implement Type-II fusion; 
$d-2$ of the input qudits are used for an ancillary state, and the remaining two qudits are the inputs to Type-II fusion. 
(The protocol of \cite{bharos2024efficienthighdimensionalentangledstate} does not use qudits for the ancillary state, but the idea is essentially equivalent.) 
The art to these fusion protocols is in choosing the appropriate ancillae and corrections in order to maximize the success probability. 
\section{Main Results: Optimal Known High-Dimensional Fusion}\label{sec:bharos extension}
We now present the optimal known protocols for fusion of Type-II linear optical qudits of dimension $d$. 
In Section~\ref{sec:bharos even}, we present the case in which $d$ is even, due to \cite{bharos2024efficienthighdimensionalentangledstate, bharoshigh}. 
In Section~\ref{sec:bharos odd}, we present our main results, extending this protocol to odd-dimensional qudits. This includes the most applicable case, in which the dimension $d$ is an odd prime. 
Finally, in Section~\ref{sec:ancilla}, we discuss an example physical system that could feasibly be used to construct the required entangled ancillary states. 

\subsection{Known even-dimensional case}\label{sec:bharos even}
We now review the fusion protocol of \cite{bharoshigh, bharos2024efficienthighdimensionalentangledstate},  applicable to qudits with even dimension $d$, with a success probability of \( 2/d^2 \). Note that when $d=2$, this recovers the standard success probability of (unboosted) Type-II fusion. 
We note that we technically follow \cite{bharoshigh}; the work \cite{bharos2024efficienthighdimensionalentangledstate} is conceptually similar, but replaces the $(d-2)$-qudit ancillary state with a related state involving $d-2$ photons within only $d$ modes. 

The circuit diagram of the protocol for $d=4$ is shown in Fig.~\ref{fig:bharos-type-II}: this is a Fourier projection with appropriately chosen ancillary state, input to the final $d-2$ ports. 
(We note that there are many equivalent choices of ancilla, discussed in \cite{bharos2024efficienthighdimensionalentangledstate, bharoshigh}, but for ease of exposition we make a specific choice.) 
In the $d=4$ case, the ancilla has $2$ photons in $8$ modes and is given by 
\begin{align}
    \ket{A_4} &= \frac{1}{\sqrt{2}}(\ket{10000100} + \ket{00100001})
    \\&= \frac{1}{\sqrt{2}}(\ket{\mathbf{01}} + \ket{\mathbf{23}}).
\end{align}
We note that, if modes $1,3,4,6$ (all of which are empty) are omitted, the state $\ket{A_4}$ is seen to be simply the \emph{two-dimensional} Bell pair $\frac{1}{\sqrt{2}}(\ket{1010} + \ket{0101})$. 
For general even $d$, we consider the following ancillary state from \cite{bharoshigh}: 
\[ \ket{A_d} = \frac{1}{\sqrt{d/2}}\sum_{r=0}^{d/2-1} \ket{\mathbf{0+2r}, \mathbf{1+2r}, \dots, \mathbf{d-3+2r}},\]
where the indices are taken modulo $d$. 
This is a state of $d-2$ linear optical qudits ($d-2$ photons across $d-2$ ports, each with $d$ time bins). The superposition will always have $d/2$ terms, and the $r$th term has qubits in all possible computational basis states except for $\ketb{d+2r-2}, \ketb{d+2r-1}$. 
For example, for $d=6$ we have
\[\ket{A_6} = \frac{1}{\sqrt{3}}\left( \ketb{0123} + \ketb{2345} + \ketb{4501} \right).\]
We discuss the construction of such ancillary states in Section~\ref{sec:ancilla}. 
Like above, the state may be simplified by omitting empty modes: namely, one can remove the odd-indexed modes from the even-indexed qudits and vice versa. 
By removing these redundant modes and applying cyclic shifts, we see that $\ket{A_d}$ is equivalent to a $(d-2)$-GHZ state in $\frac{d}{2}$ dimensions. (So $\ket{A_4}$ is a $2$-dimensional Bell pair, $\ket{A_6}$ is a $3$-dimensional $4$-GHZ state, etc.) 

\begin{protocol}\label{proto:even} For $d$ even, the protocol of \cite{bharoshigh} proceeds as follows: 
\begin{enumerate}
    \item Input arbitrary qudits in ports $0$ and $1$, and the ancillary state $\ket{A_d}$ occupying ports $2$ through $d-1$. 
    \item Permute the time bins of port $1$ according to the permutation $(01)(23)\cdots (d-2,d-1)$. 
    \item Perform a Fourier projection as in Section~\ref{sec:fourier}. 
    \item Adjust phases appropriately according to the obtained measurement pattern. 
\end{enumerate}
\end{protocol}

\begin{theorem}\label{thm:even}\cite{bharos2024efficienthighdimensionalentangledstate, bharoshigh}
    Protocol~\ref{proto:even} performs a Type-II fusion with success probability $2/d^2$. 
\end{theorem}

We now give the proof of this result, following \cite{bharos2024efficienthighdimensionalentangledstate, bharoshigh}, since we will want to extend it in the following section. 
\begin{proof}
Following Lemma~\ref{lemma:fourier projection}, if we obtain measurement pattern $(q_0, \dots, q_{d-1})$, the resulting projection is onto the state
\begin{align}
    (I_d\otimes I_d\otimes \bra{A_d})\ket{f_d^{(q)}}.
\end{align}
The $r$th term of $\ket{A_d}$ is missing the values $d+2r-2, d+2r-1$; thus the remaining states must be in the appropriate evenly weighted (but phased) superposition of $\ketb{d+2r-2, d+2r-1}$ and $\ketb{d+2r-1, d+2r-2}$. Ranging over all $0\leq r\leq d/2$, each such term appears exactly once, so there can be no interference. The projection is onto an evenly weighted superposition of 
\begin{equation}
    \ketb{01}, \ketb{10}, \ketb{23}, \ketb{32}, \dots. 
\end{equation}
For general measurement patterns $q=(q_0, \dots, q_{d-1})$, the appropriate phases may be determined and corrected using Lemma~\ref{lemma:fourier projection}. 
Up to these easily correctable phases, it suffices to consider the case $q=(0,\dots, 0)$, in which we project onto the normalized state 
\begin{equation}\label{eq:normalized projector}
    \frac{1}{\sqrt{d}}\left(\ketb{01}+\ketb{10}+\ketb{23} +\ketb{32}+\dots\right)
\end{equation}
with weight
\begin{equation}
    \frac{1}{d/2}\times \frac{1}{d^d}\times d = \frac{2}{d^d}.
\end{equation}
Here, the first factor is the scalar from $\ket{A_d}$, the second is the scalar from Lemma~\ref{lemma:fourier projection}, and the third comes from normalizing \eqref{eq:normalized projector}. 
Note, accounting for the permutation of port $1$ described in Protocol~\ref{proto:even}, we in fact project onto the ideal Bell state $\ket{B_0}$. 

Now, given a maximally mixed two-qudit input state, the probability of projecting onto $\ket{B_0}$ (one vector in a $d^2$-dimensional subspace) is $2/d^{d+2}$. 
Since there are $d^d$ successful measurement patterns $q = (q_0, \dots, q_{d-1})$, the total probability of projecting onto $\ket{B_0}$ (modulo phase corrections) is 
\begin{equation}
    \frac{2}{d^{d+2}}\times d^d = \frac{2}{d^2}.
\end{equation}
\end{proof}

\begin{figure}[htb]
    \centering
    \includegraphics[width=0.8\linewidth]{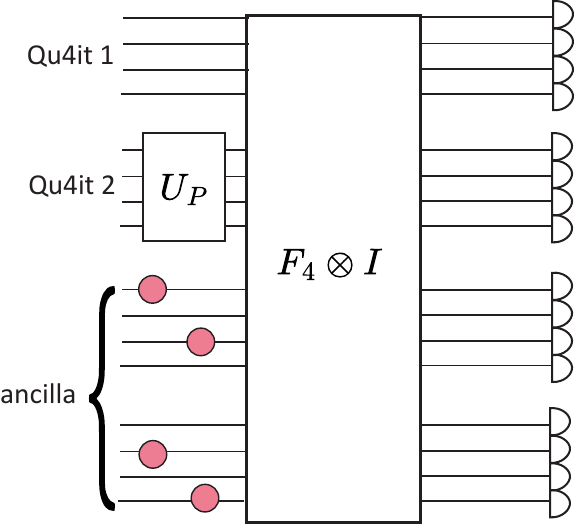}
    \caption{Type-II fusion for $d=4$ using Protocol~\ref{proto:even} \cite{bharos2024efficienthighdimensionalentangledstate, bharoshigh}. The last two input ports contain the ancilla state $\frac{1}{\sqrt{2}}(\ketb{01} + \ketb{23})$. 
    The second input qudit undergoes the permutation $U_P$ described in the protocol, swapping modes $0$ and $1$ and modes $2$ and $3$. 
    Each jth mode of the ith port undergoes a four-dimensional Fourier transform.
    }
    \label{fig:bharos-type-II}
\end{figure}
\subsection{Extension to arbitrary dimension}\label{sec:bharos odd}
We now present our main result, extending the protocol of \cite{bharoshigh} to arbitrary dimensions $d$. 
The protocol outlined in the previous section inherently requires the qudit dimension to be even, and it is unclear how to directly modify it to work with odd-dimensional qudits. 
Instead, if our dimension $d$ is odd, we simply choose some even $D>d$, embed our $d$-dimensional qudits into a $D$-dimensional subspace by adding additional vacuum modes (time bins), and apply Protocol~\ref{proto:even} for dimension $D$ (using the ancilla $\ket{A_{D}}$). 
This is most efficient when we take $D=d+1$. This protocol in the case $d=3$, $D=4$ is depicted in Fig.~\ref{fig:odd_dims}. 

\begin{protocol}\label{proto:odd}
    Let $d$ be arbitrary, and let $D\geq d$ be even. 
    We consider a state with two $d$-dimensional input qudits labeled $0$ and $1$. Our protocol for Type-II fusion is: 
    \begin{enumerate}
        \item Append $D-d$ additional vacuum modes, labeled $d, d+1, \dots, D-1$, 
        to each of the two input qudits. 
        \item Permute the time bins of the extended $D$-dimensional qudit $1$ according to the permutation $(01)(23)\cdots (D-2,D-1)$. 
        \item Perform a $D$-dimensional Fourier projection as in Section~\ref{sec:fourier}, with ancillary state $\ket{A_{D}}$. 
        \item Adjust phases appropriately according to the obtained measurement pattern. 
    \end{enumerate}
\end{protocol}

\begin{theorem}\label{thm:odd}
    For $d$ arbitrary and $D\geq d$ even, Protocol~\ref{proto:odd} is a Type-II fusion with success probability $\dfrac{2}{dD}$. In particular, for $d$ odd and $D=d+1$, we have success probability $\dfrac{2}{d(d+1)}$. 
\end{theorem}
\begin{proof}
    With the exception of the first step, this is simply the $D$-dimensional Type-II fusion protocol. 
    By the proof presented in Section~\ref{sec:bharos even}, this protocol projects onto the following unnormalized state (modulo correctable phases): 
    \begin{equation}\label{eq:bharos projector original}
        \left(\dfrac{2}{D^{D}}\right)^{1/2}\times \frac{1}{\sqrt{D}}\sum_{k=0}^{D-1}\ketb{kk}.
    \end{equation}
    We are interested in the \emph{effective projection} for input states occupying only the first $d$ modes of each qudit. 
    In other words, the effective projection is given by taking \eqref{eq:bharos projector original} and dropping any terms utilizing the computational basis states $\ketb{k}$ for $k\geq d$. 
    Doing this and reorganizing the coefficient, the effective projection is onto
    \begin{align}
        & \left(\dfrac{2}{D^{D}}\right)^{1/2}\times \frac{\sqrt{d}}{\sqrt{D}}\times \frac{1}{\sqrt{d}}\sum_{k=0}^{d-1}\ketb{kk}
        \\=& \left(\dfrac{2d}{D^{D+1}}\right)^{1/2}\ket{B_0},
    \end{align}
    where $\ket{B_0}$ is the $d$-dimensional Bell state. 
    Considering a maximally mixed input state of two $d$-dimensional qudits, the resulting probability is 
    \begin{equation}
        \dfrac{2d}{D^{D+1}}\times \dfrac{1}{d^2} = \dfrac{2}{dD^{D+1}}.
    \end{equation}
    We have $D^{D}$ successful patterns to consider, making the total success probability 
    \begin{equation}
        \dfrac{2}{dD^{D+1}}\times D^{D} = \dfrac{2}{d D}.
    \end{equation}
\end{proof}

\begin{figure}[htb]
    \centering
    \includegraphics[width=0.7\linewidth]{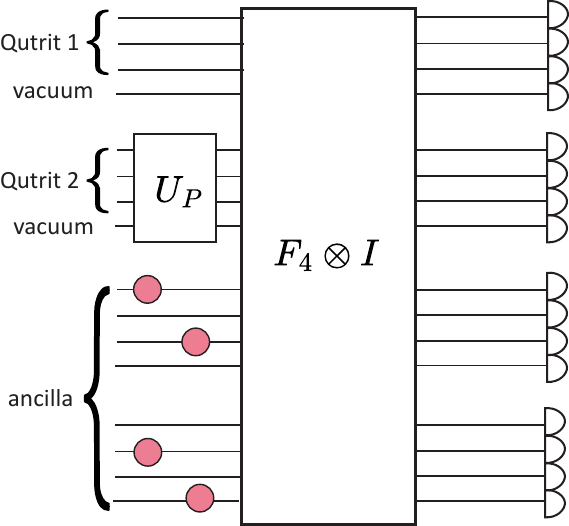}
    \caption{
    Type-II fusion for $d=3$, applying Protocol~\ref{proto:odd}. 
    The fusion protocol for odd dimensions $d$ encodes the qudits into $(d+1)$-dimensional qudits, with one extra mode in a vacuum state. After this step, the circuit and ancilla are the same as used in the $(d+1)$-dimensional case. This extends efficient Type-II fusion to odd dimensions, giving the highest known success probability for odd $d$. }
    \label{fig:odd_dims}
\end{figure}

In Table~\ref{table:odd}, we show the success probability of our protocol for different odd dimensions. 
This is the highest known success probability for every odd dimension. 
For example, with $d=5$, this \textbf{success probability is over 723 times larger} than the analogous probability from \cite{Luo} (the previous record); further, it is still over $22$ times larger than our extension of \cite{Luo} below. 
For further details, see Section~\ref{sec:luo}, especially Table~\ref{table:luo}. 

\begin{table}[ht]
\begin{tabular}{|c|c|}
\hline
$d$ & Success probability  \\ \hline
3 & 0.16                   \\ \hline
5 & 0.066                  \\ \hline
7 & 0.0357                       \\ \hline
\end{tabular}
\caption{For qudit dimensions $d=3,5,7$, we give the success probability of Protocol~\ref{proto:odd} with $D=d+1$. This protocol does not require any extra-dimensional corrections.}\label{table:odd}
\end{table}
\subsection{Constructing the ancilla}\label{sec:ancilla}

\begin{figure*}[htb]
    \includegraphics[width=.85\linewidth]{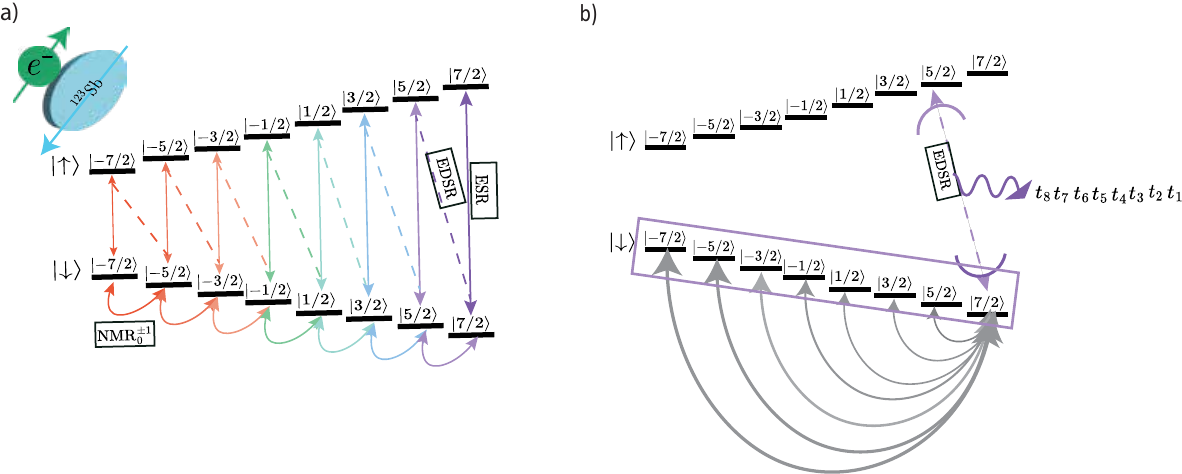}
    \caption{The energy spectrum of the neutral antimony
donor. Antimony has a high nuclear spin, with $I=7/2$. and it has eight energy levels. In its neutral charge state, the electron bound to the donor couples with the nuclear spin, expanding the Hilbert space to form a 16-level system. Curved arrows represent NMR transitions, while ESR is depicted using vertical solid arrows, and EDSR is indicated with dashed arrows. The 0 subscript in NMR emphasizes that the antimony donor is in its neutral charge state. b) Antimony donor is coupled to a microwave cavity. A time-bin protocol is used for controlling multi-mode photonic states. The microwave cavity operates at the EDSR frequency corresponding to the transition between the states $\ket{7/2}\ket{\downarrow} \leftrightarrow \ket{5/2}\ket{\uparrow}$ to enable coherent photon emission. The purple rectangle represents the uniform superposition of all nuclear spin states associated with the electron spin-down state and is used to create the proposed ancilla for qudit dimension $d=16$. The gray arrows indicate the permutation operation between nuclear spin states. $t_1$, $t_2$, $\cdots$, $t_8, \cdots$ represent the time bins into which the photon is emitted. The figure is adapted from~\cite{Ustun2025}.}
\label{fig:ancilla_construction}
\end{figure*}
\citet{bharoshigh} proposed a general method for constructing the ancilla $\ket{A_d}$ required for Protocols~\ref{proto:even} and \ref{proto:odd} using a two-level quantum emitter. 
Here, we provide an example of a silicon spin qudit, following the methods of \cite{Ustun2025}, to construct the desired ancilla in the $d=4$ case. 
Higher dimensions $d$ follow using a similar method and are discussed at the end of the section. 
Recall that, for $d$ odd, we use the even-dimensional ancilla $\ket{A_{d+1}}$, so it suffices to consider $d$ even. 
The $d=4$ example given here, then, also suffices for Protocol~\ref{proto:odd} with $d=3$. 
The required ancilla is $\ket{A_4} = \frac{1}{\sqrt{2}}(\ket{1000,0100} + \ket{0010,0001})$, recalling from Section~\ref{sec:notation} that we often use a comma to separate modes corresponding to different linear optical qudits. 
This is a superposition of $2$ terms, each involving two photons. In fact, as discussed above, the state $\ket{A_4}$ is equivalent to a dual-rail Bell pair. However, we discuss this case in detail since it readily generalizes to arbitrary $d$. 

Our proposed silicon-based quantum device utilizes the antimony donor \({}^{123}\mathrm{Sb}\), which possesses a high nuclear spin of \(I = 7/2\). This nuclear spin alone spans an 8-dimensional Hilbert space, as determined by  \(2I + 1 = 8\). As a group V element with five valence electrons, \({}^{123}\mathrm{Sb}\) behaves like a hydrogenic impurity and, at cryogenic temperatures, binds an additional electron (see Fig.~\ref{fig:ancilla_construction}a). The inclusion of this bound electron expands the total system to a 16-dimensional Hilbert space~\cite{Ustun2025}. We can control the 16-dimensional Hilbert space by applying oscillating magnetic and electric fields. For example, applying an oscillating magnetic field allows us to flip the electron spin, a process known as Electron Spin Resonance (ESR), or to change the nuclear spin projection by one quantum of angular momentum, a phenomenon referred to as Nuclear Magnetic Resonance (NMR). Alternatively, the application of an oscillating electric field induces electric dipole spin resonance (EDSR) transitions in the neutral donor. This enables transitions between anti-parallel spin states, such as \(\ket{\downarrow_e \uparrow_n} \leftrightarrow \ket{\uparrow_e \downarrow_n}\) (the dashed lines in Fig.~\ref{fig:ancilla_construction}a). The energy diagram of antimony is shown in Fig.~\ref{fig:ancilla_construction}a.

To create the desired ancilla, our objective is to encode two photons across multiple modes. To achieve this, we propose incorporating a microwave cavity into the device architecture, from which a photon is emitted using a time-bin multiplexing scheme as described in~\cite{Ustun2025}. In this scenario, the time bin into which the photon (of
a single fixed frequency) is emitted represents the designated parameter, and we can use one of the system frequencies for the microwave cavity to operate.   See Fig.~\ref{fig:ancilla_construction}b.

We start with a chosen initial state $\ket{\psi_0} = \ket{7/2}\ket{vac}\ket{\downarrow}$, where $\ket{vac}$ represents the cavity vacuum state and $\ket{\downarrow}$ represents the electron spin-down state. In the case $d=4$ (and $d=3$), we need eight time bins in total.  More specifically,
\[
\ket{A_4} = \frac{1}{\sqrt{2}}(\mid\overset{\overset{{t_1}}{\downarrow}}{1}\overset{\overset{{t_2}}{\downarrow}}{0}\overset{\overset{t_3}{\downarrow}}{0}\overset{\overset{t_4}{\downarrow}}{0}
        \overset{\overset{t_5}{\downarrow}}{0}\overset{\overset{t_6}{\downarrow}}{1}\overset{\overset{t_7}{\downarrow}}{0}\overset{\overset{t_8}{\downarrow}}{0}>+ \mid\overset{\overset{{t_1}}{\downarrow}}{0}\overset{\overset{{t_2}}{\downarrow}}{0}\overset{\overset{t_3}{\downarrow}}{1}\overset{\overset{t_4}{\downarrow}}{0}
        \overset{\overset{t_5}{\downarrow}}{0}\overset{\overset{t_6}{\downarrow}}{0}\overset{\overset{t_7}{\downarrow}}{0}\overset{\overset{t_8}{\downarrow}}{1}>),
\] where the photons only occupy four time-bins\textemdash the first, third, sixth and eighth\textemdash the others are vacuum. 
We will use two energy levels of antimony.

In the spin-down state of the electron, we create a uniform superposition of the two energy levels ($\ket{7/2}$ and $\ket{5/2}$) using a Hadamard operation. The resultant state will be 

\[\ \begin{split}
        \ket{\psi_1}=\frac{1}{\sqrt{2}}\big(& \ket{7/2} + \ket{5 /2}  \big) \ket{\downarrow}\ket{vac}.
    \end{split}\]

An EDSR pulse is then applied, which flips the spin of the electron
conditioned on the nucleus being in state $\ket{7/2}$. The state becomes:
\[\ \begin{split}
        \ket{\psi_2}=\frac{1}{\sqrt{2}} \ket{7/2}\ket{\uparrow}\ket{vac} + \ket{5 /2}  \ket{\downarrow}\ket{vac}.
    \end{split}\]
Since we use a single cavity that operates on the EDSR pulse corresponding to the transition \( \ket{7/2}\ket{\downarrow} \leftrightarrow \ket{5/2}\ket{\uparrow} \), and only this EDSR pulse is used, only the population in \( \ket{7/2} \) will be excited to the spin-up state of the electron, and not \( \ket{5/2}\ket{\downarrow} \)(See Fig.~\ref{fig:ancilla_construction}).
The electron experiences a coherent exchange of energy with the microwave cavity such that at time $t_1$, the electron is in the down state and the cavity is populated with a photon of frequency $\omega$: 
\[\ \begin{split}
        \ket{\psi_3}=\frac{1}{\sqrt{2}} \ket{7/2}\ket{\downarrow}\ket{\omega}_{t_1} + \ket{5 /2}  \ket{\downarrow}\ket{vac}.
    \end{split}\]
A permutation operation, in the case of two levels, is simply a single NMR pulse used to swap the population between nuclear states. This operation is applied between $\ket{7/2}$ and $\ket{5/2}$ with the electron in $\ket{\downarrow}$, leaving the system in the state:
\[
    \begin{split}
        \ket{\psi_4}=&\frac{1}{\sqrt{2}}\ket{5/2} \ket{\downarrow}\ket{\omega}_{t_1}+\ket{7/2}\ket{\downarrow}\ket{vac}_{t_1}
    \end{split}.
\]
The reason we permute the populations is because we are using a single cavity which operates the EDSR frequency belonging to the state $\ket{7/2}$. 

To create the desired ancilla, we create a superposition of two terms: a photon is emitted either in the first time-bin or in the third time-bin. 
At this stage, we must utilize the third mode—corresponding to the third time bin—without emitting a photon in the preceding time bin. Since we are emitting photons coherently using a cavity, everything is clocked, and we simply wait until the 3rd time bin without performing any operation. Then, in the 3rd time bin, we apply the same EDSR frequency to flip the electron spin and emit the photon coherently using the microwave cavity. This is simply a repetition of the procedure carried out during the first time bin. The resultant state then becomes:
\[
    \begin{split}
        \ket{\psi_5}=\frac{1}{\sqrt{2}}\ket{5/2} \ket{\downarrow}\ket{\omega}_{t_1}+ &\ket{7/2}\ket{\downarrow}\ket{\omega}_{t_3}  
    \end{split}
\]
\textit{Note: For simplicity, going forward we will generally omit  
the \( \ket{\text{vac}} \) states. If included, we would have}
\[
   \begin{split}
     \ket{\psi_5}=\frac{1}{\sqrt{2}}(&\ket{5/2} \ket{\downarrow}\ket{\omega}_{t_1}\ket{vac}_{t_2}\ket{vac}_{t_3}+ \\ &\ket{7/2}\ket{\downarrow}\ket{vac}_{t_1}\ket{vac}_{t_2}\ket{\omega}_{t_3} ). 
   \end{split}
\]

For the final time bin, we simply wait without pulsing the system. The resultant state has a photon in a superposition of the first and third time bins.  
Neglecting the frequency information, we may write $\ket{w_1}_{t_1}$ and $\ket{w_1}_{t_3}$ in the Fock basis, so the resultant state becomes:
\[
    \begin{split}
        \ket{\psi_6}=\frac{1}{\sqrt{2}}\ket{5/2} \ket{\downarrow}\ket{1000}+ &\ket{7/2}\ket{\downarrow}\ket{0010}  
    \end{split}
\]
We then repeat this process\textemdash applying a permutation operation, along with an EDSR pulse conditioned on the nuclear state being state $\ket{7/2}$, then emitting a photon using the microwave cavity\textemdash for time bins 6 and 8. 
The final state becomes:
\[
\begin{split}
      \ket{\psi_{final}} =&\frac{1}{\sqrt{2}}( \ket{5/2} \ket{\downarrow}\ket{1000}\ket{vac}_{t_5}\ket{\omega}_{t_6}\ket{vac}_{t_7}\ket{vac}_{t_8}+ \\ &\ket{7/2}\ket{\downarrow}\ket{0010} \ket{vac}_{t_5}\ket{vac}_{t_6}\ket{vac}_{t_7}\ket{\omega}_{t_8}) 
     \end{split},
\] 
or, written in the Fock basis,
\[
\begin{split}
     \ket{\psi_{final}} =  \frac{1}{\sqrt{2}}(\ket{5/2} \ket{\downarrow}\ket{1000,0100} +&\ket{7/2}\ket{\downarrow}\ket{0010,0001})
     \end{split}.
\]
We then apply a Hadamard operation between the states of antimony to obtain
\[
\begin{split}
     \ket{\psi_{final'}} =  & \frac{1}{2}(\ket{7/2}-\ket{5/2})\ket{\downarrow}\ket{1000,0100}+ \\ &  \frac{1}{2}(\ket{7/2} + \ket{5/2})\ket{\downarrow}\ket{0010,0001}).
     \end{split}
\]
If we now measure the antimony, decoupling it from the photons, we will project onto one of the two states
\[
\begin{split}
\ket{\psi_{aux}} = \frac{1}{\sqrt{2}}\left(\ket{1000,0100} + \ket{0010,0001}\right) = \ket{A_4}
\end{split}
\]
or
\[
\begin{split}
\ket{\psi_{aux}} = \frac{1}{\sqrt{2}}\left(\ket{1000,0100} - \ket{0010,0001}\right),
\end{split}
\] where the phase is heralded and can be corrected as usual. 

This method is only a small version of the protocol which is studied in~\cite{Ustun2025}.
The Hamiltonian of the presented physical system can be found in~\ref{app:Hamiltonian}. 
For qudit dimension $d=6$ (and $d=5$), the ancilla consists of a four-photon state in 24 time bins. Each grouping of $6$ modes (corresponding to a qudit) can have photons in one of three different time bins. In this situation, we should use three energy levels of antimony such as  $\ket{7/2}$, 
$\ket{5/2}$, and $\ket{3/2}$, instead of using only two energy levels. For qudit dimensions $d > 4$, the permutation operation can either be performed by sequential NMR pulses on the antimony or by using more sophisticated control schemes, such as the global rotation described in~\cite{Yu2025}.

In Fig.~\ref{fig:ancilla_construction}b, a more general use of this protocol is shown for a qudit dimension of $d = 16$. In this case, 224 time bins are required, and in each term, each qudit has a photon in one of eight different time bins. 
This necessitates a uniform superposition over all the ground states, while the remaining modes remain in the vacuum state. 

For qudit dimensions $d > 16$, one can create a uniform superposition using spin-down and spin-up states of the electron together. In theory, for a single antimony donor, it is possible to create an ancilla for qudit dimensions up to $d = 28$, where the ancilla would require 14 photons across 728 time bins. This is, of course, a theoretical hypothesis that disregards any experimental limitations, such as the coherence time of antimony or the time scale of the permutation operation. In practice, the creation of such an ancilla state on current antimony hardware would likely be a challenge even for $d = 6$. 
While this system offers a high degree of control, is well understood, and has already been used in several experiments~\cite{Fern_ndez_de_Fuentes_2024, Yu2025}, it is not necessarily the optimal platform for this purpose. 
At present, we aim only to present a concrete and physically realizable system for the construction of the ancillary states $\ket{A_d}$ required for Protocol~\ref{proto:odd}. 

\subsection{Feasibility of the Protocol}\label{sec:implement}
We discuss the feasibility of near-term implementations of Protocols~\ref{proto:even} and \ref{proto:odd} through three key considerations: detector requirements, circuit complexity, and ancillary entanglement resources.
\begin{enumerate}
    \item \textbf{Detector Requirements:} We begin by examining the assumptions related to PNRD. For arbitrary qudit dimension, the protocols presented in this section only require detectors capable of distinguishing among three categories: vacuum ($0$ photons), single-photon events, and multi-photon events ($2$ or more photons). 
    This is because the Fourier projection (Protocol~\ref{proto:fourier}) post-selects for measurement outcomes in which each mode contains at most one photon. 
    In the absence of multiphoton errors, so that there can be at most $d$ photons present in the system, it is sufficient to distinguish between vacuum and non-vacuum events. 
    This requirement mirrors that of fusion-based protocols in the qubit setting and is within the reach of current experimental capabilities. For example, recent works~\cite{simone2025, hauser2025boosted} employ a standard pseudo-PNRD approach to demonstrate high-fidelity \emph{boosted} (qubit) fusion, which imposes stricter detection requirements than those needed for the protocols of this section.
    
    \item \textbf{Circuit Complexity:} Next, we evaluate the complexity of the optical circuits. The principal challenge lies in photon loss, which becomes increasingly problematic with higher photon counts and reduces the probability of successful circuit execution. For qudit dimensions $d = 3,4$ (respectively, $d=5,6$), the required interferometric circuit can be implemented using a mesh of Mach-Zehnder interferometers involving $4$ photons in $16$ modes (respectively $6$ photons in $36$ modes) \cite{barak2007quantum, clements2016optimal}. The case of $4$ photons is comfortably within the current experimental capabilities; for example, QuiX Quantum has demonstrated a $20$ mode reconfigurable linear optical processor with low loss, high visibility, and high-fidelity operations~\cite{smith2022universal}. Although loss limits the total number of usable photons, $4$-photon experiments are routinely performed, as in~\cite{simone2025, hauser2025boosted}, where boosted qubit teleportation was achieved using $4$--$5$ photons in $8$ modes. The $6$-photon case is also considered feasible with current or near-term hardware. For example, as early as 2016, a $10$-photon entangled state was experimentally realized with sufficient fidelity to certify genuine multipartite entanglement~\cite{wang2016experimental}.

    \item \textbf{Ancillary Entanglement Resources:} Finally, we consider the generation of the ancillary entangled states required. Section~\ref{sec:ancilla} outlines a method for generating the necessary ancillae using a $d$-level quantum emitter. In particular, the antimony donor system discussed here and in~\cite{Ustun2025} is realistically capable of generating such states up to dimension $d=8$, and potentially up to $d=16$ with improved coherence properties. Near-term feasibility is expected for dimensions up to $d \leq 6$. Experimental control of antimony-based systems has been demonstrated in related work~\cite{Fern_ndez_de_Fuentes_2024,Yu2025}.

    For small qudit dimensions, alternative ancilla-generation strategies are available. As discussed in Section~\ref{sec:bharos even}, in the cases of $d=3$ and $d=4$, the required ancilla reduces to a standard \emph{two-dimensional} linear optical Bell pair; these have been routinely generated in experiments for decades, e.g. \cite{kwiat1995new, michler1996interferometric}. 
    
\end{enumerate}

\section{Type-II Fusion Protocols Using Extra-Dimensional Corrections}\label{sec:all_extra_dims}
We begin by introducing the notion of extra-dimensional corrections, generalizing methods of \cite{Luo} by applying techniques for Procrustean distillation as in \cite{vidal1999entanglement}. 
Similar ideas are considered in \cite{Paesani} in the context of Bell state generation. 

In Section~\ref{sec:luo}, we review the protocol of \cite{Luo}, then extend it and improve its performance using extra-dimensional corrections. 
Afterward, we consider other variants of Type-II fusion that exploit extra-dimensional corrections. These are generally able to improve the success probability of the relevant circuits a great deal. 
All circuits presented here perform worse than our Protocol~\ref{proto:odd} above (or its even-dimensional analogue, Protocol~\ref{proto:even}, due to \cite{bharoshigh}). 
However, these protocols have more than theoretical interest as negative results. 
One major advantage is that the protocols of this section do not require multi-photon entangled states as ancillae: only single photons, W-states, and single-mode bunched states $\ket{k}$ are utilized. 

\subsection{Extra-Dimensional Corrections}\label{sec:USD}
We now consider circuits for Type-II fusion (or quantum teleportation) that do not project onto a Bell state such as $\ket{B_0}$, but rather a state of the form 
\begin{equation}\label{eq:uncorrected state}
    \ket{\Psi} = \frac{1}{\sqrt{d}}\sum_{k=0}^{d-1} \ketb{k}\ket{\psi_k},
\end{equation}
where the $\ket{\psi_k}$ are linearly independent but \emph{not necessarily orthogonal} vectors. 
(Note we do not require the $\ket{\psi_k}$ to be normalized either.) 
In linear optics, where the qudit dimension is easily increased by appending additional modes, we have the flexibility to ``correct" such a state into a Bell pair by performing a POVM on a higher-dimensional space, as described shortly. 
(This is standard in circuits for boosted fusion and entangled linear optical state generation, e.g. Figs. \ref{fig:boosted_type-II}, \ref{fig:boosted_type-II-2}, \ref{fig:paesani}.) 
For practical application of this method in fusion or teleportation, one cannot directly apply such a correction, as the relevant photons are destroyed during the measurement. 
Instead, a related correction is applied on the \emph{other} qudits entangled with the second input qudit. 
We discuss the practical applicability of this method in Section~\ref{sec:corrections_applied}, for now just assuming we want to correct \eqref{eq:uncorrected state} into a Bell state. 

The desired corrections can be shown to exist using the Schmidt decomposition and a theorem of \citet{vidal1999entanglement}: 
\begin{theorem}\label{thm:vidal}
    \cite{vidal1999entanglement} Let $\ket{\Psi}$ and $\ket{\Phi}$ be two-qudit states (of local dimension $d$) with Schmidt decompositions
    \begin{equation}
        \ket{\Psi} = \sum_{i=0}^{d-1} \lambda_i \ket{\alpha_i}\ket{\beta_i},\,\,\, \ket{\Phi} = \sum_{i=0}^{d-1} \mu_i \ket{\gamma_i}\ket{\delta_i},
    \end{equation}
    where we take the $\lambda_i$ and $\mu_i$ to be weakly increasing. 
    The optimal probability of converting $\ket{\Psi}$ into $\ket{\Phi}$ using only local POVMs is given by
    \begin{equation}
        \min_{t=1,\dots, d-1} \dfrac{\sum_{i=0}^{t-1}\lambda_i^2}{\sum_{i=0}^{t-1}\mu_i^2}.
    \end{equation}
\end{theorem}
Taking $\ket{\Phi}$ in the theorem to be a Bell pair, we have all $\mu_i^2 = 1/d$, and we obtain 
\begin{corollary}
    A two-qudit state with Schmidt decomposition $\ket{\Psi} = \sum_i \lambda_i \ket{\alpha_i}\ket{\beta_i}$ can be corrected into a Bell pair using local POVMs with probability $d\lambda_0^2$, where $d$ is the qudit dimension and $\lambda_0$ is the smallest Schmidt coefficient. 
\end{corollary}

In our setting, one may directly use the singular value decomposition (which underlies the Schmidt decomposition) to explicitly construct the POVM required to correct the state \eqref{eq:uncorrected state}. 
Further, it suffices to act on only the second qudit. 
Naively, one would like to apply the transformation
\begin{equation}\label{eq:matrixA}
    A = \sum_k \ketb{k}\bra{\psi_k^\perp},
\end{equation}
where the $\bra{\psi_i^\perp}$ are a dual basis defined by $\braket{\psi_i^\perp}{\psi_j} = \delta_{ij}$. 
(Note $A$ is simply the inverse of the matrix whose columns are the $\ket{\psi_k}$; this is well-defined because we assume $\{\ket{\psi_i}: 0\leq i\leq d-1\}$ is an independent set.) 
This satisfies $(I\otimes A)\ket{\Psi} = \ket{B_0}$, and $A$ is invertible. However, it is not unitary in general because the $\ket{\psi_k}$ are not necessarily orthogonal (or unit vectors). 
However, we may extend (a nonzero multiple of) $A$ to a unitary by increasing the qudit dimension, as follows: 
\begin{theorem}\label{thm:filip corrections}
    Let $B = (\braket{\psi_j}{\psi_i})_{ij}$ be the Gram matrix, with smallest eigenvalue $\lambda$. 
    Let $s\leq d-1$ be the number of eigenvalues of $B$ that are \emph{not} equal to $\lambda$. 
    If we augment the second qudit with $s$ extra dimensions (modes), there exists a unitary $U$ extending the matrix $\sqrt{\lambda} A$, with 
    \begin{equation}
        (I\otimes U)\ket{\Psi} = \frac{\sqrt{\lambda}}{\sqrt{d}}\sum_k \ketb{kk} + \ket{\textnormal{junk}}.
    \end{equation}
    Here, in each term of the ``junk" state, the second qudit's photon is always in one of the newly added modes. 
    Performing PNRD on the $s$ new modes, we post-select for the case in which no photons are detected; this occurs with probability $\lambda$. 
\end{theorem}
This is proven in Appendix~\ref{app:USD}. 
Viewing $A$ as a $d\times d$ matrix, we first extend an appropriate multiple of $A$ to a $(d+s)\times d$ isometry $U'$, then use Gram-Schmidt to obtain the unitary $U$. 
An example is given in Section~\ref{sec:derivation}. 

Note that if we let $M_k = U\ketbra{\psi_k}U^\dagger$, then the POVM given by $M_0, \dots, M_{d-1}, I - \sum_k M_k$ may be used for unambiguous discrimination of the states $\ket{\psi_k}$. In particular, obtaining outcome $M_k$ implies that the state was $\ket{\psi_k}$, while obtaining $I- \sum_k M_k$ is inconclusive. 
\subsubsection{Practical Extra-Dimensional Corrections}\label{sec:corrections_applied}
We first consider the context of teleportation, where \citet{Luo} initially applied a version of extra-dimensional corrections. 
In that setting, Alice has a state $\ket{\eta}$ (qudit $0$), Alice and Bob share a Bell pair $\ket{B_0}$ (qudits $1$ and $2$), and we perform the projection $\bra{\Psi}$ determined by \eqref{eq:uncorrected state} on qudits $0$ and $1$. 
The resulting state is 
\begin{equation}
    (\bra{\Psi}\otimes I_d)\ket{\eta}\ket{B_0}.
\end{equation}
Recall that we have a non-unitary matrix $A$ with $(I\otimes A)\ket{\Psi} = \ket{B_0}$. 
Noting that 
\begin{equation}
    (F\otimes I)\ket{B_0} = (I\otimes F^T)\ket{B_0}
\end{equation}
for any single-qudit linear transformation $F$, we may insert the matrix $A$ by applying $\overline{A} = (A^\dagger)^T$ to qudit $2$: 
\begin{align}
    &(I\otimes I \otimes \overline{A})(\bra{\Psi}\otimes I_d)\ket{\eta}\ket{B_0}
    \\=& (\bra{\Psi}\otimes I_d)(I\otimes I \otimes \overline{A})\ket{\eta}\ket{B_0}
    \\=& (\bra{\Psi}\otimes I_d)(I\otimes A^\dagger \otimes I)\ket{\eta}\ket{B_0}
    \\&= (\bra{B_0}\otimes I_d)\ket{\eta}\ket{B_0}.
\end{align}
The final expression is equivalent to the state $\ket{\eta}$, supported on qudit $2$ as in standard teleportation. 
Then \emph{after} the projection \eqref{eq:uncorrected state} is determined, one should use Theorem~\ref{thm:filip corrections} to extend $\overline{A}$ to a unitary and apply it to qudit $2$. The relevant unitary will simply be the complex conjugate of the matrix $U$ constructed above, and the probability $\lambda$ will be unchanged. Then the success probability corresponding to $\ket{\Psi}$ is $\lambda \times p$, where $p$ is the probability associated with the projector $\ket{\Psi}$. 
This is spelled out in detail in Section~\ref{app:general fusion}. 

This justifies the use of extra-dimensional corrections in teleportation, or more broadly any Type-II fusion in which one of the states to be fused is a maximally entangled pair. 
In a more general setting, in which both states involved in the fusion exhibit many-qudit entanglement, Theorem~\ref{thm:filip corrections} is not directly applicable. Instead, one may often apply the more general Theorem~\ref{thm:vidal} to correct the resulting state using local POVMs. 
We have not proven that the overall success probability remains the same using such methods, so the success probabilties of the present Section~\ref{sec:all_extra_dims} should be interpreted as upper bounds in the most general context. 
We further note that implementing extra-dimensional corrections will require increased circuit depth and active switching (or active reconfiguration of the circuit) depending on the obtained measurement outcome. This will likely increase the rate of photon loss in the circuit, further reducing the practical success probability. 

\subsection{Type-II fusion with W-state ancillae}\label{sec:luo}

We review the protocol of \cite{Luo}, put it into the context of extra-dimensional corrections, and finally present an extension that greatly improves the success probability. 

\subsubsection{Original protocol}\label{sec:luo original}

We now consider the teleportation protocol of \cite{Luo}, which in fact corresponds to a Type-II fusion protocol as discussed above. 
\citet{Luo} present the protocol in greatest detail for the qutrit case $d=3$, with a partial generalization to arbitrary dimensions presented. 
Let 
\begin{equation}
    \ket{W_d} = \frac{1}{\sqrt{d}}\sum_{k=0}^{d-1}\ketb{k}
\end{equation}
be the $d$-dimensional W-state, the higher-dimensional analogue of the $\ket{+}$ state. 

\begin{figure}[htb]
    \centering
    \includegraphics[width=.9\linewidth]{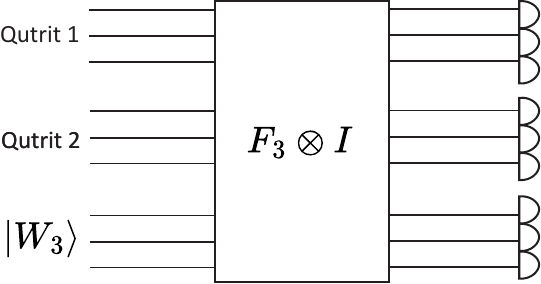}
    \caption{Type-II fusion for $d=3$, following \cite{Luo}. Each triplet in the circuit represents a port (qudit), and each port has three modes. The last three modes —the last port— belong to the ancilla state $\ket{W_3}$. For every $0\leq j\leq d-1$, we perform a three-dimensional Fourier transform involving the $j$th modes of each port. Concretely, we apply Fourier transforms to modes $(0,3,6), (1,4,7), (2,5,8)$. 
    }
    \label{fig:luo-fusion}
\end{figure}
In the case $d=3$, the protocol of \cite{Luo} begins with a Fourier projection as in Protocol~\ref{proto:fourier}, where two input qudits are inserted into ports $0$ and $1$, and port $2$ is occupied by an ancillary W-state, $\ket{W_3}$. 
(See Fig.~\ref{fig:luo-fusion}.) 
However, the effective projection on the input qudits is \emph{not} a Bell measurement, or even a projection onto a maximally entangled state. 
Instead, the projection is onto a state of the form \eqref{eq:uncorrected state}. 
Thus, as discussed in Section~\ref{sec:USD}, an extra-dimensional correction must be applied. 
In this $3$-dimensional case, all such corrections require only $1$ extra mode, and the appropriate unitaries were calculated to have the form 
\begin{equation}\label{eq:luo 3d correction}
    U = \frac{1}{2}\begin{pmatrix}
    -1 & 1 & 1 & 1\\ 1 & -1 & 1 & 1\\ 1 & 1 & -1 & 1\\ 1 & 1 & 1 & -1
\end{pmatrix},
\end{equation}
up to phases determined by the specific measurement pattern \cite{Luo}. 
(The derivation of this correction operator using the results of Section~\ref{sec:USD} is given in Section~\ref{sec:derivation}.) 
The resulting success probability is $1/9$. 
Although it is tempting to conjecture that the $d$-dimensional generalization will have probability $1/d^2$ (and this is stated as fact in the literature), we see below that this is not the case. 

For arbitrary $d$, \citet{Luo} present a generalization: perform Fourier projection using the $(d-2)$-qudit ancillary state $\ket{W_d}^{\otimes (d-2)}$, then perform Fourier projection, but only treating measurement patterns as successful if all photons are detected in the same port (and different time bins, as usual). 
For these $d$ patterns (one for each port), a single extra dimension still suffices to enable extra-dimensional projections, and the appropriate $(d+1)$-dimensional unitary is given in \cite{Luo} as a direct generalization of \eqref{eq:luo 3d correction}. For this protocol, the resulting success probability may be explicitly calculated as
\begin{equation}\label{eq:luo success prob}
    \left(\dfrac{(d-2)!}{d^{d-1}}\right)^2\sim \dfrac{2\pi}{d e^{2d}},
\end{equation}
where the asymptotic follows from Stirling's approximation. 
(Note this formula gives probability $1/81$ in the $d=3$ case, rather than $1/9$ as discussed above, since only $3$ of the $27$ possible patterns are considered.) 
We see that the success probabilities decay exponentially with the qudit dimension. 
\subsubsection{Derivation of Correction Matrix for $d=3$}\label{sec:derivation}
In this section, we derive the correction matrix \eqref{eq:luo 3d correction} \cite{Luo}, following the concepts and notation of Section~\ref{sec:USD}, especially Theorem~\ref{thm:filip corrections}. 
This section is simply meant as an example illustrating the method of extra-dimensional corrections and relating it to the work of \cite{Luo}. 
The state we wish to correct to a Bell state has the form \eqref{eq:uncorrected state}, 
where \[ \ket{\psi_0}= \frac{1}{\sqrt{2}}(\ket{1} + \ket{2}) = \frac{1}{\sqrt{2}} \begin{pmatrix} 0 \\ 1 \\ 1 \end{pmatrix}, \] 
\[\ket{\psi_1}= \frac{1}{\sqrt{2}}(\ket{0} + \ket{2}) = \frac{1}{\sqrt{2}}\begin{pmatrix} 1 \\ 0 \\ 1 \end{pmatrix}, \] 
\[\ket{\psi_2} = \frac{1}{\sqrt{2}}(\ket{0} + \ket{1}) = \frac{1}{\sqrt{2}}\begin{pmatrix} 1 \\ 1 \\ 0 \end{pmatrix}. \] 
We now consider the matrix $\phi$ whose columns are the $\ket{\psi_k}$:
\[\phi =\frac{1}{\sqrt{2}} \begin{pmatrix}
        0 &1& 1 \\
        1 &0 &0 \\
        1 &1& 0
    \end{pmatrix}
\]
Note that $\phi$ is the inverse of the matrix called $A$ in Section~\ref{sec:USD}. 
We invert $\phi$ to obtain
\[
A=\frac{1}{\sqrt{2}}\begin{pmatrix}
        -1&1&1\\
        1&-1&1\\
        1&1&-1
    \end{pmatrix}.
\]
The matrix $B = (AA^\dagger)^{-1}$ is:
\[
 B=\frac{1}{2} \begin{pmatrix}
        2&1&1\\
        1&2&1\\
        1&1&2
    \end{pmatrix}.
\]
The eigenvalues of $B$, with multiplicity, are $\{2, 1/2, 1/2\}$. The smallest eigenvalue is $\lambda=1/2$, giving the correction probability. 
The number of extra dimensions is $s=1$, since $B$ has only one eigenvalue that is not equal to $\lambda$. 
In particular, letting 
\[ M= \sqrt{\lambda} A = \frac{1}{2}\begin{pmatrix}
        -1&1&1\\
        1&-1&1\\
        1&1&-1
    \end{pmatrix},\]
we will embed $M$ into a unitary $U$ of dimension $4$, which we now construct. 
Note $M$ is normalized to have maximum singular value $1$, so that we can use Lemma~\ref{lemma:usd} to embed $M$ into an isometry $\begin{bmatrix}
        M\\ S
    \end{bmatrix}$. 
We calculate the singular vector of $M$ corresponding to its smallest singular value $1/2$ to be $(1/\sqrt{3})\begin{pmatrix} 1 & 1 & 1\end{pmatrix}^T$. 
Following the lemma, $S$ will be a $1\times 3$ matrix with the above singular vector and singular value $\sqrt{1-(1/2)^2} = \sqrt{3}/2$. 
We then have
\[S = \frac{\sqrt{3}}{2}\frac{1}{\sqrt{3}}\begin{pmatrix} 1 & 1 & 1\end{pmatrix} = 
\frac{1}{2}\begin{pmatrix} 1 & 1 & 1\end{pmatrix},\]
leading to the following isometry and corresponding unitary: 
\[
\begin{bmatrix}
        M\\ S
    \end{bmatrix} = \frac{1}{2}\begin{pmatrix}
    -1 & 1 & 1\\ 1 & -1 & 1\\ 1 & 1 & -1\\ 1 & 1 & 1
\end{pmatrix}, \,\,\, U = \frac{1}{2}\begin{pmatrix}
    -1 & 1 & 1 & 1\\ 1 & -1 & 1 & 1\\ 1 & 1 & -1 & 1\\ 1 & 1 & 1 & -1
\end{pmatrix}.
\]
The final column of $U$ follows from Gram-Schmidt.

\subsubsection{Improved W-state protocol using extra-dimensional corrections}\label{sec:scaling}
We now discuss an improved version of the protocol of Section~\ref{sec:luo original} \cite{Luo} and compare the protocols' performance. 
Note that in the $d$-dimensional case above ($d>3$), only $d$ measurement patterns were accepted in the Fourier projection, instead of the usual $d^d$ patterns. 
Patterns with photons in different ports still lead to a projection onto the state given in Lemma~\ref{lemma:fourier projection}, but not all terms have the same phase. 
When projecting onto the $W$-state ancillae, these differing phases lead to nontrivial destructive interference. 
Thus, unlike the cases considered in Section~\ref{sec:bharos extension}, not all measurement patterns may be treated equally and may have wildly differing associated probabilities. 
One may still perform Type-II fusion, correcting the resulting projections using extra-dimensional corrections, but many will need more than $1$ additional mode, and the correction unitaries vary greatly, not directly generalizing \eqref{eq:luo 3d correction}. 
We compare the success probabilities of these two cases in Table~\ref{table:luo}. The first column gives the success probability when following the method of \citet{Luo} for general $d$, numerically verifying the formula \eqref{eq:luo success prob}. (Note that the $d=3$ case does \emph{not} match the value of $1/9$ given in \cite{Luo}, because in that one case, the method of \citet{Luo} matches our own.) The second column gives our method for increasing the success probability using extra-dimensional corrections, and the final column gives the average number of extra modes required. 
We observe that all values are worse than those obtained using Protocols~\ref{proto:even} and \ref{proto:odd} above (see Table~\ref{table:odd} and Theorems~\ref{thm:even}, \ref{thm:odd}). 
As noted above, our Protocol~\ref{proto:odd} is over $22$ times more efficient than even the improved W-state protocol for $d=5$, and in fact the gap only widens as $d$ increases. Further, Protocol~\ref{proto:odd} \emph{does not} require extra-dimensional corrections, making it both more efficient and simpler to implement (although requiring the preparation of an entangled multi-qudit state). 

\begin{table}[ht]
\begin{tabular}{|c|c|c|c|}
\hline
$d$ & Success probability-1 & Success probability-2 & Average $s$ \\ \hline
3 & 0.012 & 0.111               & 1.0                                \\ \hline
4 & $9.8\times 10^{-4}$ & 0.017               & 2.229                              \\ \hline
5 & $9.2\times 10^{-5}$ & 0.003               & 3.685                              \\ \hline
\end{tabular}
\caption{For qudit dimensions $d=3,4,5$, we give the success probability of the protocol of \cite{Luo} and our generalization, where the ancillae are $d-2$ single-qudit $W$ states $\ket{W_d}$. In the column labeled ``Success probability-1," we give the total success probability when \emph{only} patterns with all photons in the same port (and different modes) are accepted. 
This is the method of \cite{Luo} for $d>3$. 
In the column labeled ``Success probability-2," we give the total success probability when one applies extra-dimensional corrections whenever possible. 
The final column records the average number $s$ of additional dimensions required for the extra-dimensional corrections in the latter case. 
The raw data may be found in the ancillary files on arXiv. 
\label{table:luo}}
\end{table}

\subsection{Extension of State Generation Circuits for Fusion}\label{sec:paesani}

\begin{figure}[htb]
    \centering
    \includegraphics[width=\linewidth]{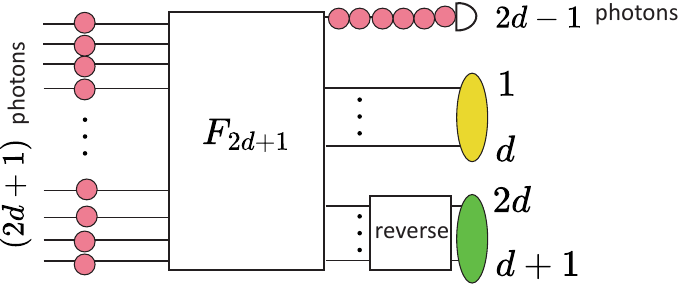}
    \caption{This Fourier transform-based scheme \cite{Paesani} enables the heralded generation of Bell states in arbitrary dimensions $d$. 
    The relationship between the output modes and the corresponding computational states of each qudit is illustrated. 
    The ``reverse" gate indicates that the order of the modes of the second qudit should be reversed to obtain a Bell state. 
    }
    \label{fig:paesani}
\end{figure}

\begin{figure}
    \includegraphics[width=.8\linewidth]{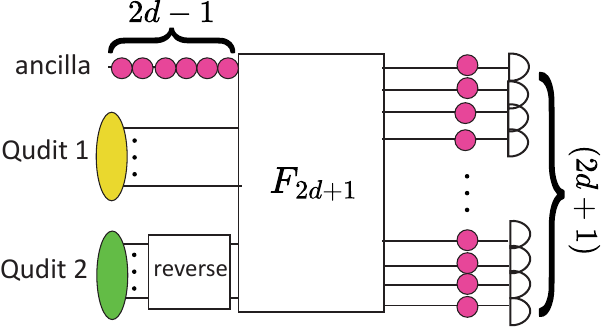}
    \caption{A Type-II fusion protocol suggested by reversing the Bell state generation circuit of \citet{Paesani}. 
    This requires an ancillary bunched state of $2d-1$ photons in a single mode as an ancilla, as well as two $d$-dimensional qudits as input. 
}
    \label{fig:paesani_all}
\end{figure}

\begin{table}[ht]
\begin{tabular}{|c|c|c|c|c|c|}
\hline
      & $r=1$ & $r=2$ & $r=3$ & $r=4$ & $r=5$ \\ \hline
$d=3$ & 0.116 & 0.116 & 0.109 & 0.140 & 0.136 \\ \hline
$d=4$ & 0.0   & 0.020 & 0.038 & 0.047 & 0.053 \\ \hline
$d=5$ & 0.0   & 0.0   & 0.004 & 0.011 & 0.018 \\ \hline
\end{tabular}
\caption{
Fusion success probabilities (including extra-dimensional corrections) for the ZTL circuit of Fig.~\ref{fig:paesani_all}, in qudit dimension $3\leq d\leq 5$ and using $1\leq r\leq 5$ ancillary photons. 
In this case, nearly all successful measurement patterns require extra-dimensional corrections involving $d-1$ additional modes.
The raw data may be found in the ancillary files on arXiv. }
\label{table:paesani}
\end{table}

In this section, we consider a potential Type-II fusion gate that does not use the Fourier projection of Section~\ref{sec:fourier}, although it still crucially uses the Fourier transform. 
Specifically, we consider the Bell state generation circuit introduced by \citet{Paesani} and depicted in Fig.~\ref{fig:paesani}. 
This circuit relies on the Zero Transmission Law (ZTL), which describes powerful (but not exhaustive) suppression laws for boson sampling using the Fourier transform \cite{tichy2010zero, saied2025general}. 
Thus it is referred to as the \emph{ZTL circuit}. 
In this circuit, $2d+1$ photons are input into different modes of a $F_{2d+1}$ interferometer, the $0$th mode is measured using PNRD, and one post-selects for finding $2d-1$ photons in that mode. 
The output on the remaining $2d$ modes is a $d$-dimensional Bell state $\ket{B_0}$, up to reversing the order of the modes of one of the qudits. 

As discussed in Section~\ref{sec:teleportation}, one can generally convert a Bell state generation circuit into a Type-II fusion. We do this in Fig.~\ref{fig:paesani_all}. 
To perform this variant of Type-II fusion, one inputs an ancillary bunched state of $2d-1$ photons in the $0$th mode, and the two input qudits occupy the remaining $2d$ modes. The modes of the second input qudit are reversed, then a $F_{2d+1}$ Fourier interferometer is applied. 
The natural analogue of the Bell state generation protocol would then post-select for the pattern $\ket{1, \dots, 1}$ with $1$ photon measured in each output mode, as depicted in Fig.~\ref{fig:paesani_all}. 
This does in fact give a Bell state projection, with relatively low success probability
\begin{equation}
    \dfrac{(2d-1)!}{d(2d+1)^{2d-1}},
\end{equation}
which follows directly from a result of \cite{Paesani}. 
To improve the success probability, we allow arbitrary post-selection patterns and the corresponding extra-dimensional corrections, weighting each appropriately. 
Further, given the flexibility of extra-dimensional corrections, one may consider different single-mode ancillary states $\ket{r}$ rather than only $\ket{2d-1}$. 
The corresponding success probabilities are given for $d=3,4,5$ and $1\leq r\leq 5$ in Table~\ref{table:paesani}, calculated numerically. 
This data shows many interesting patterns. 
First, note that a larger number of ancillary photons is not necessarily better: in dimension $d=3$, the case $r=4$ leads to a higher success probability than $r=5$. 
Further, the success probability is $0$ unless $r\geq d-2$; this aligns with the fusion protocols above, in which the ancillary states involve at least $d-2$ photons. 
However, for $r=d-2$ (and many larger values of $r$), we find that the success probability is always \emph{greater} than that of \cite{Luo} and our extension: compare with Table~\ref{table:luo}. 
The tradeoff is that, while the protocols of this section are more efficient than those of Section~\ref{sec:luo}, the extra-dimensional corrections are more complex, with $d-1$ additional modes generally required. 
Further, these protocols require preparation of bunched states $\ket{k}$, which may be challenging in practice. 
Of course, we must also recall that none of the protocols under discussion here have success probability as high as Protocols~\ref{proto:even} and \ref{proto:odd}, which require entangled ancillae but avoid the need for extra-dimensional corrections. 

In Appendix~\ref{app:paesani}, we consider Type-II fusion gates derived from other Bell state generation protocols of \cite{Paesani}, but all are numerically found to have very low success probability even with extra-dimensional corrections. 

\subsection{Novel boosted fusion for qubit case using extra-dimensional corrections}\label{sec:psiq-increased}

\begin{figure}[ht]
    \centering
    \includegraphics[width=0.5\linewidth]{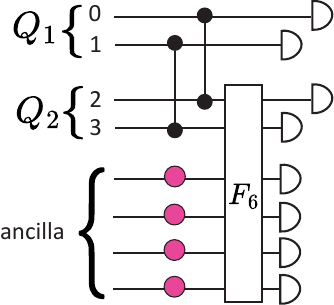}
    \caption{An extension of the circuit defined in~\cite{bartolucci2021creationentangledphotonicstates}. We add four modes, each containing a single photon, and perform a Fourier transform of the appropriate size. 
    As seen in Table~\ref{table:psiq}, adding more photons and modes in this way increases the overall success probability (if one uses extra-dimensional corrections). 
    }
    \label{fig:novel_boosting}
\end{figure}
We now return to the protocol of \cite{bartolucci2021creationentangledphotonicstates} for boosted Type-II fusion of \emph{qubits}, presented in Fig.~\ref{fig:boosted_type-II-2}. 
This protocol is notable because it does not generally project onto a stabilizer state, but rather requires (non-Clifford) unitary corrections. 
Thus it is a natural candidate for generalization with extra-dimensional corrections. 
(Also note that at most one extra mode will be required for these corrections, by Section~\ref{sec:USD}.) 
We consider a generalization of this protocol, still in the case $d=2$, but increasing the success probability with additional ancillary photons and extra-dimensional corrections. 
Specifically, for any $r\geq 1$, we replace each $F_3$ with an $F_{2+r}$, acting on two modes from the relevant qubit and $r$ ancillary modes, each occupied by a single photon. This generalization is depicted for $r=4$ in Fig.~\ref{fig:novel_boosting}, where we apply the boosting technique on only one side for simplicity of presentation. 
We present the success probabilities of the resulting Type-II fusion protocols, with and without extra-dimensional corrections, in Table~\ref{table:psiq}. 
As usual, these numbers may be increased by applying the same treatment to the other pair of modes. 
For small numbers of ancillary photons, these success probabilities are not particularly high: for example, recall that the boosting protocol of Fig.~\ref{fig:boosted_type-II} can obtain success probabilities of $0.625$ and $0.75$ with $2$ and $4$ single photons respectively \cite{Ewert_2014}. 
For the protocol presented here, however, the numerics seem to imply that \emph{the success probability strictly increases with the number of ancillary single photons}. 
(Note this is different from the statistics of Table~\ref{table:paesani}.) 
If this trend holds, this may lead to a useful method for high-probability Type-II fusion, since it does not require entangled ancillae. 

\begin{table*}[htb]
\begin{adjustbox}{width=.8\linewidth}
\begin{tabular}{|c|c|c|c|c|}
\hline
\# of extra photons & Success Probability without Corrections&Success probability with Corrections &   Average $s$  \\ \hline
1 (original case)& 0.583 & 0.583 & 0           \\ \hline
2 & 0.578&0.593               & 0.051          \\ \hline
3 & 0.575 &0.606               & 0.067         \\ \hline
4 & 0.586&0.620               & 0.063          \\ \hline
5 & 0.590&0.630               & 0.059          \\ \hline
\end{tabular}
\caption{Increasing the success probability for the qubit case ($d=2$) using additional single photons as in Fig.~\ref{fig:novel_boosting}. Only when using extra-dimensional corrections (Sec.~\ref{sec:USD}) does the success probability continue to increase as more ancillary photons are used. The raw data may be found in the ancillary files on arXiv. \label{table:psiq}}
\end{adjustbox}
\end{table*}

It is natural to attempt the same strategy in qudit dimension $d>2$. 
For example, we may take the circuit of Fig.~\ref{fig:luo-fusion} and, before measurement, apply a Fourier transform involving the modes of one output port and $r$ ancillary single photons. 
However, this seems to make the success probabilities strictly worse: see Table~\ref{table:psiq qutrit} for numerics in the $d=3$ case. 

\begin{table}[ht]
\begin{adjustbox}{width=1\linewidth}
\begin{tabular}{|c|c|c|c|}
\hline
\# of extra photons & Success probability with Corrections &   Average $s$ \\ \hline 
1 & 0.077 & 1.805 \\ \hline
2 & 0.071 & 1.881 \\ \hline
3 & 0.076 & 1.803 \\ \hline
\end{tabular}
\caption{In the qutrit case $d=3$, the analogue of the boosting protocol of Fig.~\ref{fig:novel_boosting} performs poorly. 
The use of ancillary single photons numerically seems to always decrease the success probability vs. the baseline probability $1/9$ obtained by \cite{Luo}.  The raw data may be found in the ancillary files on arXiv. 
\label{table:psiq qutrit}
}
\end{adjustbox}
\end{table}
\section{Conclusions and Outlook}\label{sec:conclusion}

\begin{table*}[htb]
\centering
\resizebox{\linewidth}{!}{ 
\begin{tabular}{|c|*{10}{c|}}
\hline
 & \multicolumn{3}{c|}{$d=3$} & \multicolumn{3}{c|}{$d=4$} & \multicolumn{3}{c|}{$d=5$} \\
\cline{2-10}
& Ancilla Type & EDC used & Success Probability & Ancilla Type & EDC used & Success Probability & Ancilla Type & EDC used & Success Probability \\
\hline
\citet{bharos2024efficienthighdimensionalentangledstate,bharoshigh} & N/A & N/A & N/A & Entangled $2$ photon & No & 0.125 & N/A & N/A & N/A \\
\hline
Sec. \ref{sec:bharos odd} & Entangled 2 photon& No & 0.16 & Entangled $2$ photon & No & 0.125 & Entangled 4 photon& No & 0.067 \\
\hline
\citet{Luo} & W state & No & 0.111 & 2 W states & No & $9.8\times 10^{-4}$  & 3 W states & No & $9.2 \times 10^{-5}$ \\
\hline
Section \ref{sec:scaling} & W state & Yes & 0.111 & 2 W states & Yes & 0.017 & 3 W states & Yes & 0.003 \\
\hline
Sec. \ref{sec:paesani} (few-photon version) & 1 Single photon & Yes & 0.111 & Bunched 2 photon & Yes & 0.020 & Bunched 3 photon & Yes & 0.004 \\
\hline
Sec. \ref{sec:paesani} (high-probability version) & Bunched 4 photon & Yes & 0.140 & Bunched 5 photon & Yes & 0.056 & Bunched 5 photon & Yes & 0.018 \\
\hline
\end{tabular}
}
\caption{We summarize many of the Type-II fusion protocols considered in the present work for $3\leq d\leq 5$. For each protocol and each dimension, we recall the type of ancillary state used, whether extra-dimensional corrections (EDC) are required, and the success probability. 
In the first three rows, the probabilities are known exactly; in the final three rows, the probabilities are calculated numerically. 
For the protocols of Sec.~\ref{sec:paesani}, we choose only two variants extracted from Table~\ref{table:paesani}. 
\label{table:summary conclusion}
}
\end{table*}
Type-II fusion, and the corresponding notion of teleportation, is a cornerstone concept with widespread applications across many domains of physics. 
In this work, we focus specifically on the role of Type-II fusion in enabling high-dimensional quantum computing. 
While several variants of high-dimensional linear-optical fusion have been proposed \cite{Luo, bharos2024efficienthighdimensionalentangledstate, bharoshigh}, they remain largely underexplored. 
We address this gap by analyzing, improving, and extending existing fusion methods, giving what we believe to be the first efficiently-scaling Type-II fusion protocol for odd-dimensional qudits. 
In the $5$-dimensional case, this method exhibits a $723$-fold improvement over previous work \cite{Luo}, and the gap only widens as the dimension $d$ increases. 

We first analyze the method of \cite{bharos2024efficienthighdimensionalentangledstate, bharoshigh}, which achieves Type-II fusion with success probability $2/d^2$ for even dimensions $d$. 
In Section~\ref{sec:bharos odd}, Protocol~\ref{proto:odd}, we adapt this method to obtain an odd-dimensional Type-II fusion gate with success probability $2/d(d+1)$, which we believe to be the most effective known protocol for Type-II fusion in odd dimensions. Furthermore, we propose a physical implementation pathway for the required ancillary state. This state may be constructed from a spin qudit in silicon, coupled to a microwave cavity using a time-bin multiplexing scheme. This is an alternative version of the method in~\cite{Ustun2025}. 

The method we present in Protocol~\ref{proto:odd} works by embedding qudits of odd dimension $d$ into a larger even-dimensional space. 
This is a natural technique in linear optics, where photons can naturally spread across many modes, and is similar in spirit to other methods of ``boosted" fusion that use ancillary modes and photons. 
This flexibility of qudit dimension in linear optics naturally leads to the notion of \emph{extra-dimensional corrections} presented in Section~\ref{sec:USD}, which can allow circuits that project onto many \emph{non}-maximally-entangled states to be non-deterministically corrected into Bell measurements. 
This technique works by implementing a unitary on a larger space involving a target qudit and newly added vacuum modes, then measuring the new mode and post-selecting for the case in which no photons are detected there. This is a concrete linear optical special case of the work of \cite{vidal1999entanglement} and generalizes a technique already used to obtain a Type-II fusion gate in \cite{Luo}. 

We then consider various protocols utilizing extra-dimensional corrections. 
We investigate and extend the method of \cite{Luo}, improving its success probability. 
Despite claims in the literature, we note in Section~\ref{sec:luo} that both the original method of \cite{Luo} and our extension have success probability significantly worse than $1/d^2$ for general $d$. 
The success probability is given analytically for the original protocol and shown numerically for our extension. 
We also consider Type-II fusion gates derived from the high-dimensional Bell state generation circuits of \cite{Paesani}. 
(See Sec.~\ref{sec:paesani} and App.~\ref{app:paesani}.) 
Utilizing extra-dimensional corrections, these can outperform our extension of \cite{Luo}, but still have worse success probability than Refs. \cite{bharos2024efficienthighdimensionalentangledstate, bharoshigh} (even dimensions) and our Protocol~\ref{proto:odd} (odd dimensions). 

We also consider an application of extra-dimensional corrections to extend a boosted fusion circuit of \cite{bartolucci2021creationentangledphotonicstates} in the \emph{qubit} case. 
We find a circuit that performs worse than existing boosting protocols such as \cite{Grice, Ewert_2014} for small numbers of ancillary photons, but numerically seems to monotonically increase in success probability as more single photons are added. 
This circuit may be useful for practical implementations of Type-II fusion (for qubits) with very high success probability, since it requires no entangled ancillae. 

In Appendix~\ref{sec:boosted_high_dim}, we also evaluate a boosted fusion protocol directly generalizing the method of \cite{Ewert_2014} and relate it to the zero-transmission law for the Fourier transform \cite{tichy2010zero, saied2025general}. 
This method is primarily of theoretical interest, as the improvement over the protocols of Section~\ref{sec:luo} is very small and requires a large number of ancillary photons. 
(Further, the success probability is still worse than that of Protocols~\ref{proto:even} and \ref{proto:odd}.)

To summarize the various Type-II fusion protocols considered here, we note the following. 
(Also see Tables~\ref{table:summary} and \ref{table:summary conclusion}.) 
The work of \cite{bharos2024efficienthighdimensionalentangledstate, bharoshigh} and our extension (Protocols~\ref{proto:even} and \ref{proto:odd} respectively) obtain the highest known success probabilities in all dimensions and \emph{do not} require extra-dimensional corrections, but \emph{do} require a highly entangled $(d-2)$-qudit ancillary state (equivalent to a $d/2$-dimensional GHZ state) as input. 
(Recall the construction of these states discussed in Section~\ref{sec:ancilla}.) 
The other protocols we consider generally require extra-dimensional corrections, whose implementation may be quite complicated in practice (see Section~\ref{sec:corrections_applied}), and have lower success probability, but also require only single photons or bunched states as ancillae. 
Among this family, our extensions of the work of \cite{Paesani} given in Section~\ref{sec:paesani} seem to have the highest success probability. 
At present, this is only known numerically, but this protocol can likely be better understood via the theory of the Zero Transmission Law and related suppression laws for the discrete Fourier transform \cite{tichy2010zero, saied2025general}. 

We also recall that the protocols of \cite{Luo, bharos2024efficienthighdimensionalentangledstate, bharoshigh} and our generalizations in Sections~\ref{sec:bharos odd} and \ref{sec:scaling}, are simply variants of the Fourier projection (see Section~\ref{sec:fourier}). This operation may be viewed as a projection onto a symmetric entangled state involving $d$ qudits of dimension $d$; this state is \emph{not} a stabilizer state for $d>2$. 
The protocols we consider occupy $d-2$ of the input qudits with ancillae, turning the Fourier projection into a two-qudit measurement. 
It is possible that one can exploit the Fourier projection as a \emph{many-body} entangling measurement in a nontrivial manner, allowing for the creation of larger-scale entanglement with fewer steps and higher success probability than obtainable by iterating fusions of the type considered here. 
A variant of this is utilized in the boosted fusion circuit of Appendix~\ref{sec:boosted_high_dim}, which can be viewed as projecting onto a $d$-qudit GHZ state. 
We leave this direction for future work. 

Another avenue for future work relates to the possibility for high-dimensional fault-tolerant quantum computing. 
Recall the example of the $[[5,1,3]]_{\mathbb{Z}_d}$ modular-qudit code \cite{chau1997five}, which uses $5$ physical ($d$-dimensional) qudits to encode one $d$-dimensional logical qudit and can tolerate $2$ erasures. 
Similarly, the qudit surface code \cite{bullock2007qudit} is $[[n^2, 1, n]]_{\mathbb{Z}_d}$, using $n^2$ physical ($d$-dimensional) qudits to encode one $d$-dimensional logical qudit, with the ability to tolerate any $d-1$ erasures. 
These examples imply that qudits may be able to store more information with the same robustness to error using the same number of particles. 
Thus linear optical qudits may potentially be useful for high-dimensional fault-tolerant quantum computation, such as analogues of FBQC, which have not yet been developed in the qudit case. 
If developed, an important question will be whether the fusion success probabilities we obtain here are sufficient for high-dimensional FBQC, and for which values of $d$. 
Although the $d>2$ case has lower success probability than the qubit case, 
it is possible that more efficient error correcting codes will make up for this difference for small $d$, or perhaps more efficient ``boosted'' fusions can be developed to make these cases feasible. 

Linear optical quantum computing is a promising and potentially highly scalable platform. Through this work, we aim to advance the understanding and practical implementation of high-dimensional fusion, an essential tool in measurement-based approaches to high-dimensional linear optical quantum computation. While this is by no means the final word, our results lay a strong foundation for future research and open new pathways for further improvements in success probability and resource efficiency.

\section*{Acknowledgments}
 We thank Filip Maciejewski and Namit Anand for fruitful discussions and useful comments on the theory of extra-dimensional corrections.
G.\"U. acknowledges supports from Sydney Quantum Academy (SQA) where she is a primary scholarship holder. This research was developed with funding from the Defense Advanced Research Projects Agency [under the Quantum Benchmarking(QB) program under award no. HR00112230007 and HR001121S0026 contracts]. 
We are grateful for support from  DARPA under IAA 8839, Annex 130, and from  NASA Ames Research Center. 
The United States Government retains, and by accepting the article for publication, the publisher acknowledges that the United States Government retains, a nonexclusive, paid-up, irrevocable, worldwide license to publish or reproduce the published form of this work, or allow others to do so, for United States Government purposes.
\newpage
\section*{Appendix}
\setcounter{section}{0}
\setcounter{equation}{0}
\setcounter{figure}{0}
\renewcommand{\thesection}{A-\Roman{section}}
\renewcommand{\theequation}{A.\arabic{equation}}
\renewcommand{\thefigure}{A.\arabic{figure}}
\section{Calculating fusion success probabilities}\label{app:general fusion}
We now consider a more general perspective of Type-II fusion between a pair of $d$-dimensional qudits. A Type-II fusion protocol involves: 
\begin{enumerate}
    \item Two states $\ket{\chi_0}, \ket{\chi_1}$, each sending one $d$-mode qudit as input to the fusion. For the fusion to be useful, the $\ket{\chi_i}$ are generally multi-qudit entangled states, with the other qudits not involved in the fusion. 
    \item A (possibly vacuum) ancillary state $\ket{\textnormal{anc}}$ on $N$ modes. Typically the ancilla will be a state of $d-2$ $d$-dimensional qudits, but this is not required. 
    \item A linear optical unitary $\widetilde{U}$ (often a variant of the Fourier transform), operating on the $N+2d$ \emph{active modes} (the two input qudits and the $N$-mode ancilla). 
    \item Photon-number-resolving detection on the $N+2d$ active modes, performed after the unitary. 
    \item A list of \emph{good} measurement patterns. If a measurement pattern is not good, the fusion is considered a failure. 
    \item Correction POVMs, adaptively chosen based on the obtained (good) measurement pattern. In the general picture, variants of the correction POVMs are performed on the qudits of (say) $\ket{\chi_1}$ that are \emph{not} directly involved in the fusion. Note that in many cases, such as \cite{Terry_2005, bharos2024efficienthighdimensionalentangledstate}, the corrections are unitaries used simply to eliminate unwanted phase factors. 
    In other cases, such as \cite{Luo} and Section~\ref{sec:all_extra_dims}, not all outcomes of the correction POVM may lead to success. See the discussion in Section~\ref{sec:USD}. 
\end{enumerate}
The outcome of such a protocol is an effective Bell measurement on the input qudits of the state $\ket{\chi_0}\ket{\chi_1}$. The input qudits are destroyed as part of the measurement, but entanglement is created between the other qudits of the two states. 

To calculate the success probability of a Type-II fusion protocol, we take each good pattern $\ket{p}$ and calculate
\[
\sqrt{w_p}\ket{p'} = \bra{Q_0} \bra{Q_1}\bra{\textnormal{anc}}\widetilde{U}^\dagger \ket{p},
\]
where $\bra{Q_0}$ is the projector onto states with exactly $1$ photon in the first $d$ modes, and similar for $\bra{Q_1}$. 
Then $\ket{p'}$ is a normalized two-qudit state, and $w_p\leq 1$ is the appropriate weight (possibly $0$, in which case $\ket{p'}$ is arbitrary and we stop the computation). 
We express $\ket{p'}$ in the form \eqref{eq:uncorrected state} and let $\lambda_p$ be the resulting correction factor from Theorem~\ref{thm:filip corrections}, with $\bra{P_p}$ the projection onto the corrected subspace. 
Then by the theorem, we obtain
\begin{equation}
    \sqrt{w_p\lambda_p} \bra{P_p} \ket{p'} 
    = \sqrt{w_p\lambda_p}\ket{B_0},
\end{equation}
where $\ket{B_0}$ is the ideal Bell state. 
In other words, with probability $w_p\lambda_p$, the (corrected) fusion procedure with measurement pattern $\ket{p}$ projects onto $\ket{B_0}$. 
Then the success probability for arbitrary $\ket{\chi_0}, \ket{\chi_1}$ is obtained by summing the overlap over all good patterns $\ket{p}$: 
\begin{equation}
    \sum_p w_p\lambda_p \left|\bra{B_0}(\ket{\chi_0}\ket{\chi_1})\right|^2.
\end{equation}
To express this as a success probability for \emph{random} input states, as is typical, we replace all $\left|\bra{B_0}(\ket{\chi_0}\ket{\chi_1})\right|^2$ with $1/d^2$, with $d^2$ being the dimension of the two-qudit subspace. 
In other words, we replace $\ket{\chi_0}\ket{\chi_1}$ with a maximally mixed two-qudit state $\frac{1}{d^2}I$, which satisfies $\Tr\left(\frac{1}{d^2}I\ketbra{B_0}\right) = \frac{1}{d^2}$. 
Then the success probability for random input states is
\begin{equation}
    \frac{1}{d^2}\sum_p w_p\lambda_p.
\end{equation}
This assumption, that the success probability is for \emph{random} input states, is rarely stated but important. For example, the classic Type-II fusion circuit \cite{Terry_2005} (see Fig.~\ref{fig:fusion-types}) is generally stated to have success probability $1/2$, but in fact it will always succeed for certain Bell pair input states and always fail for others. 
(Compare with \cite{Ewert_2014}, where this observation is explicitly used in designing boosted fusion circuits.) 

\section{Hamiltonian of the system for constructing the ancilla}\label{app:Hamiltonian}
The details of the system's Hamiltonian are the same as in~\cite{Ustun2025} and can also be found there.
For time-bin multiplexing, we use a single microwave cavity fabricated on our silicon chip, which is coupled to an antimony donor. This cavity is designed to operate at the EDSR frequency, specifically between the states  $\ket{\downarrow}\ket{7/2} \leftrightarrow \ket{\uparrow}\ket{5/2}$. The total Hamiltonian of the system is given by:
\begin{equation}
    H_{\text{total}} = H_{Sb} + H_{\text{interaction}} + H_{\text{field}}
\end{equation}
where $H_{Sb}$ is the Hamiltonian of the donor which is defined in Eq~1 in main text. $ H_{\text{interaction}} + H_{\text{field}}$ part represents the emitting photon through cavity. The total Hamiltonian can be expressed as follows:
\begin{widetext}
\begin{equation}\label{drift_Ham}
    \begin{aligned}
        H_{\text{total}} =  \underbrace{ B_0(-\gamma_n \hat{I_z} + \gamma_e \hat{S_z}) + A(\Vec{S}\cdot\Vec{I}) + \Sigma_{\alpha, \beta \in \{x,y,z\}} Q_{\alpha \beta}\hat{I}_\alpha \hat{I}_\beta}_{H_{Sb}} \\
        + \left (\underbrace{ \hbar w_{\ket{7/2}\leftrightarrow\ket{5/2}} a^{\dag} a}_{H_{\text{field}}}       
        +   \underbrace{g (\ket{7/2}\ket{\downarrow}\bra{5/2}\bra{\uparrow} a^{\dag}) + g(\ket{5/2}\ket{\uparrow}\bra{7/2}\bra{\downarrow} a)}_{H_{\text{interaction}}} \right)\\       
    \end{aligned}
\end{equation}
\end{widetext}where the operators $a^{\dag}$ and $a$ represent creation and annihilation operators, respectively. $g$ represents the coupling strength between the spin and the cavity. When purely magnetic coupling is employed via the ESR transitions, $g$ is in the range of $10$Hz - $100$kHz~\cite{Vine2023latched}. However, when utilizing an electric dipole for spin-cavity coupling, significantly higher coupling strengths -- in the range of a few MHz\cite{tosi_2017, Tossi_2014} -- and thus faster photon emission times are possible. This is why we propose to operate the cavity at the antimony donor EDSR transitions.

The first term in $H_{Sb}$ accounts for the Zeeman splitting, on both the electron and the nuclei. The second term describes the hyperfine interaction that arises from the overlap of electron and nuclear wavefunctions. In the last term, where $\alpha, \beta = {x, y, z}$ represent Cartesian axes, $\hat{I}_\alpha$ and $\hat{I}_\beta$ are the corresponding 8-dimensional nuclear spin projection operators. The term $Q_{\alpha\beta} = e q_n V_{\alpha\beta}/2I(2I-1)h$ represents the nuclear quadrupole interaction energy, determined by the electric field gradient (EFG) tensor $V_{\alpha\beta} = \partial^2V (x, y, z)/\partial\alpha\partial\beta$ \cite{asaad2020coherent}. This quadrupole interaction introduces an orientation-dependent energy shift to the nuclear Zeeman levels, enabling the individual addressability of nuclear states. $B_0$ represents the magnetic field in which the nanoelectronic device containing the antimony donor is placed, with a value approximately equal to $1 \text{T}$.
This ensures that the eigenstates of $H_{\text{Sb}}$ are approximately the tensor products of the nuclear states $\ket{m_I}$ with the eigenstates $ \{\ket{\downarrow} , \ket{\uparrow} \} $ of $\hat{S}_z$ because $\gamma_e B_0 \gg A \gg Q_{\alpha\beta}$. The latter condition implies $H_{\text{Sb}} \approx B_0(-\gamma_n \hat{I_z} + \gamma_e \hat{S_z}) + A(\Vec{S} \cdot \Vec{I})$ ensuring that the nuclear spin operator approximately commutes with the electron-nuclear interaction.
This condition allows for an approximate quantum non-demolition (QND) readout of the nuclear spin via the electron spin ancilla \cite{joecker2024error}.

\section{Proof of Theorem~\ref{thm:filip corrections}}\label{app:USD}
We recall the setting of Section~\ref{sec:USD}. We will prove Theorem~\ref{thm:filip corrections}. 
First we need the following lemma: 
\begin{lemma}\label{lemma:usd}
    Let $M$ be a $d\times d$ matrix with maximum singular value equal to $1$; let $s$ be the number of singular values less than $1$. 
    Then $M$ can be extended to a $(d+s)\times d$ isometry $\overline{M} = \begin{bmatrix}
        M\\ S
    \end{bmatrix}$. 
    Further, any isometry extending $M$ in this way must have at least $d+s$ rows. 
\end{lemma}
\begin{proof}
    For $\overline{M}$ to be an isometry, we must have
    \begin{equation}
        I = \overline{M}^\dagger \overline{M} = M^\dagger M + S^\dagger S.
    \end{equation}
    We will construct $S$ with 
    \begin{equation}\label{eq:lemma s condition}
        S^\dagger S = I - M^\dagger M.
    \end{equation}
    Write the singular value decomposition of $M$ as $M=WDV$, where $D = \textnormal{diag}(m_1, \dots, m_d)$ and we assume the $m_i$ are in increasing order. 
    Let $\ket{v_i}$ be the singular vector corresponding to $m_i$. 
    Let $D' = \textnormal{diag}(\sqrt{1-m_1^2}, \dots, \sqrt{1-m_s^2})$, an $s\times s$ diagonal matrix, and let $V'$ be the $d\times s$ matrix whose columns are the first $s$ singular vectors of $M$. 
    Then $S = D'(V')^\dagger$ satisfies
    \begin{align}
        S^\dagger S &= V' (D')^\dagger D' (V')^\dagger
        \\&= V' \textnormal{diag}(1-m_1^2, \dots, 1-m_s^2)(V')^\dagger
        \\&= \sum_{i=1}^s (1-m_i^2)\ket{v_i}\bra{v_i}
        \\&= \sum_{i=1}^d (1-m_i^2)\ket{v_i}\bra{v_i}
        \\&= I - \sum_{i=1}^d m_i^2\ket{v_i}\bra{v_i}
        \\&= I-M^\dagger M.
    \end{align}
Then $S$ suffices to extend $M$. 
Further, comparing the ranks in \eqref{eq:lemma s condition} shows that $s$ is minimal. 
\end{proof}

We now show how to construct the required extra-dimensional correction matrices, extending the matrix $A$ of \eqref{eq:matrixA}. 
We restate Theorem~\ref{thm:filip corrections} for convenience: 
\begin{theorem}
    Let $B = (\braket{\psi_j}{\psi_i})_{ij}$ be the Gram matrix, with smallest eigenvalue $\lambda$. 
    Let $s\leq d-1$ be the number of eigenvalues of $B$ that are \emph{not} equal to $\lambda$. 
    If we augment the second qudit with $s$ extra dimensions (modes), there exists a unitary $U$ extending the matrix $\sqrt{\lambda} A$, with 
    \begin{equation}
        (I\otimes U)\ket{\Psi} = \frac{\sqrt{\lambda}}{\sqrt{d}}\sum_k \ketb{kk} + \ket{\textnormal{junk}}.
    \end{equation}
    Here, in each term of the ``junk" state, the second qudit's photon is always in one of the newly added modes. 
    Performing PNRD on the $s$ new modes, we post-select for the case in which no photons are detected; this occurs with probability $\lambda$. 
\end{theorem}
\begin{proof}
We have 
\begin{equation}\label{eq:AB relationship}
    B = (A^{-1})^\dagger A^{-1}.
\end{equation}
Let $\lambda$ be the minimum eigenvalue of $B$, necessarily positive since $A$ is invertible. 
By \eqref{eq:AB relationship}, $\sqrt{\lambda}$ is the minimum singular value of $A^{-1}$, and therefore $1/\sqrt{\lambda}$ is the maximum singular value of $A$. 
Let $M=\sqrt{\lambda}A$, which has maximum singular value $1$. 
By Lemma~\ref{lemma:usd}, $M$ can be extended to a $(d+s)\times d$ isometry $\begin{bmatrix}
        M\\ S
    \end{bmatrix}$. 
This extends to a $(d+s)\times (d+s)$ unitary $U$ using Gram-Schmidt. 
The unitary $U$ satisfies
\begin{align}
    U\ket{\psi_k} &= M\ket{\psi_k} + S\ket{\psi_k} 
    \\&= \sqrt{\lambda}\ket{k} + \textnormal{junk},
\end{align}
where the ``junk" term has its photon in the newly added dimensions. 
In particular, 
\begin{equation}
    (I\otimes U)\frac{1}{\sqrt{d}}\sum_k \ket{k}\ket{\psi_k} =  \frac{\sqrt{\lambda}}{\sqrt{d}}\sum_k \ket{kk} + \textnormal{junk}.
\end{equation}
Performing PNRD on the $s$ extra dimensions of the second qudit, the probability of measuring no photons (so that we project onto a valid two-qudit state) is $\lambda$. 
\end{proof}
Note that Lemma~\ref{lemma:usd} is constructive, using the singular value decomposition. Then this proof gives a constructive algorithm for the extra-dimensional unitary correction $U$. 

\section{Extension of State Generation Circuits for
Fusion - 2}\label{app:paesani}
\begin{table}[htb]
\centering
\begin{adjustbox}{width=1\linewidth}
\begin{tabular}{|c|c|c|c|}
\hline
 & \multicolumn{3}{c|}{Success probability with corrections ($d=3$)}  \\ \hline
Circuit 1 & $U_A$: 0.0185 \hspace{0.3cm} & $U_B$:0.0185 \hspace{0.3cm} & $U_C$:0.0185  \\ \hline
Circuit 2 & \multicolumn{3}{c|}{0.0078}  \\ \hline
\end{tabular}
\end{adjustbox}
\caption{Fusion Success Probabilities with circuits inspired by~\cite{Paesani} using the theorem in Sec:~\ref{sec:USD}. The raw data may be found in the ancillary files on arXiv. }
\label{table:paesani-app}
\end{table}
\begin{figure*}
    \centering
    \includegraphics[width=.8\linewidth]{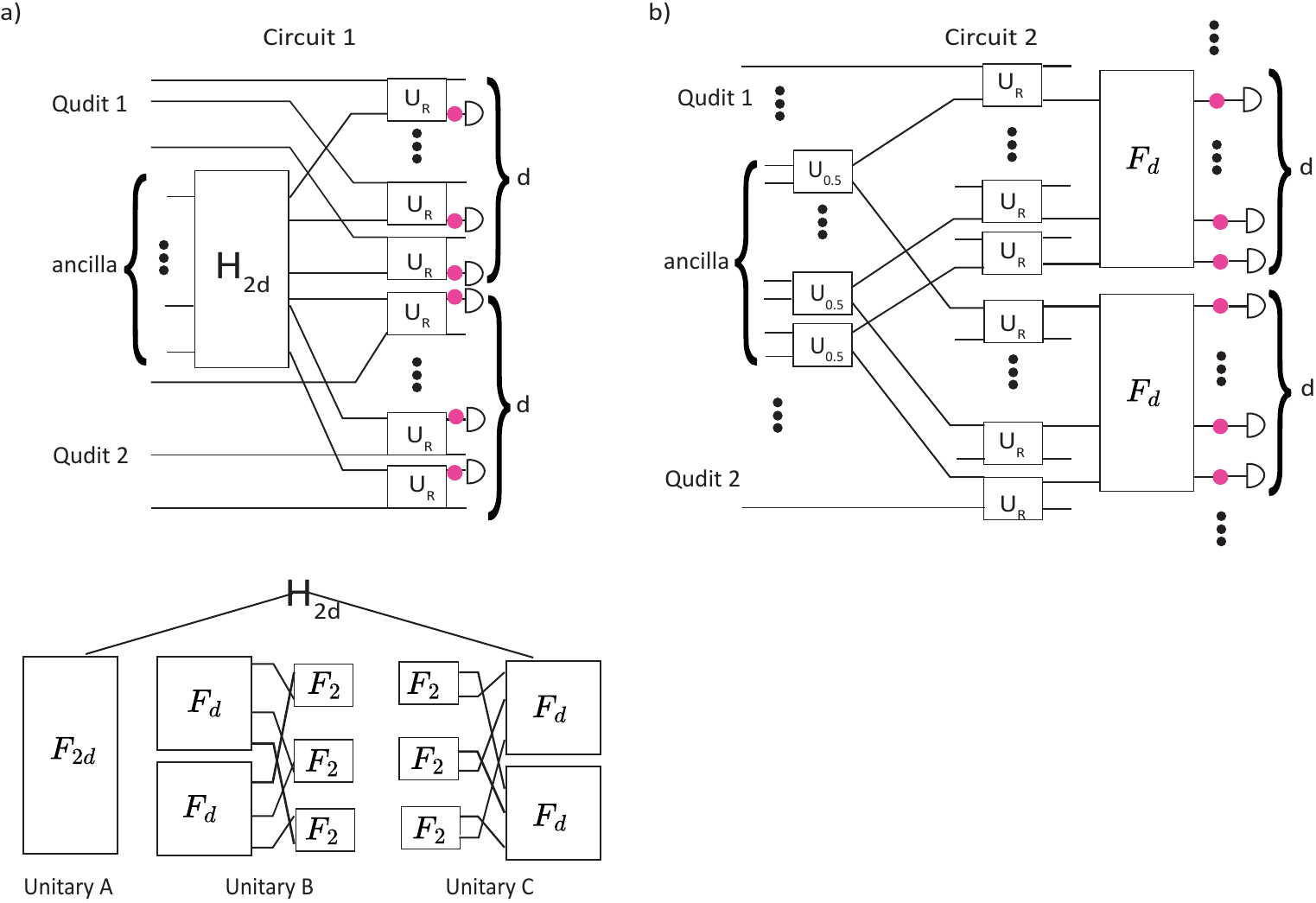} 
    \caption{We present several $d$-dimensional fusion protocols by reversing the high-dimensional Bell state generation circuits of \citet{Paesani}. a) This circuit requires a $2d$-mode ancilla. The unitary operation $H_{2d}$ can be implemented in three different ways: i) as a single $2d$-dimensional discrete Fourier transform (DFT), ii) as two $d$-dimensional DFTs followed by $d$ two-dimensional DFTs (equivalent to Hadamard operations), iii) or as $d$ two-dimensional DFTs (Hadamard operations) followed by two $d$-dimensional DFTs.
    b) Circuit~2 requires three distinct operations: $U_{0.5}=H$ (a 50:50 beamsplitter), $U_R$, and two $d$-dimensional DFTs. Note that the vacuum modes in Circuit~2 are not drawn explicitly.}
    \label{fig:Paesani-appendix}
\end{figure*}

In this section, we adapt the Bell state generation circuits considered in the Supplemental Material of \cite{Paesani} into Type-II fusion circuits. 
These circuits are presented in Fig.~\ref{fig:Paesani-appendix}. Circuits~(1) and~(2) involve the application of a parameterized gate $U_R$, defined as:
\[
   U_R = \begin{pmatrix} \sqrt{R} & i\sqrt{1 - R} \\ 
                i\sqrt{1 - R} & \sqrt{R}
\end{pmatrix}
\]
We simulated the performance of these circuits, in the case $d=3$, for 100 different values of $R$, uniformly spaced between 0 and 1 with an interval of 0.01. In all cases, the highest fusion success probability was numerically obtained at $R = 0.5$, which differs from the optimal value for the state generation circuits, where the best performance was achieved at $R = 2/3$.

In Circuit~1, the unitary operation $H_{2d}$ represents a variant of $2d$-dimensional Fourier transform: we call the three cases $U_A, U_B, U_C$, as depicted in the figure. 
While these implementations exhibited different performance levels during state generation, they numerically yielded identical success probabilities in the corresponding Type-II fusion operations. 

We note that, to leverage the structure of the state generation protocol, the ancillary state should be a superposition of the post-selection patterns used in the corresponding state generation. 
In the $d=3$ case, we numerically tested various evenly-weighted superpositions of these patterns, finding that the ancillary state 
\begin{equation}
    \frac{1}{\sqrt{2}}\left(\ket{200020} + \ket{020200}\right)
\end{equation}
led to optimal results for Circuit 1,
and 
\begin{equation}
    \frac{1}{\sqrt{2}}\left(\ket{000202} + \ket{002020}\right)
\end{equation}
led to optimal results for Circuit 2. 
(We note that many choices of ancilla lead to the same success probability, and which choices of ancilla are viable do sometimes depend on the circuit and choice of $H_{2d}$.) 
A summary of our results for the $d = 3$ case, using those ancillary states, is provided in Table~\ref{table:paesani-app}. 
The success probabilities did not surpass those of the ZTL protocol derived from the same work \cite{Paesani} and discussed in the main text, Sec.~\ref{sec:paesani}. 

\section{Boosted high-dimensional Type-II fusion}\label{sec:boosted_high_dim} 
In this section, we present a boosted high-dimensional Type-II fusion protocol combining the methods of ~\cite{Ewert_2014} and Section~\ref{sec:luo}. 
This protocol will require a great number of ancillary photons for a very small increase in success probability; however, we present it here due to interesting theoretical features. 
It may be viewed as an extension of the protocol of Section~\ref{sec:luo}: the projectors of that case still occur with the same probabilities, but we convert some ``failure'' outcomes into successful entangling measurements. 
In fact, these new outcomes may be viewed as projecting onto the $d$-qudit \emph{GHZ state} rather than a $2$-body Bell state. 

\begin{figure}[htb]
    \centering
    \includegraphics[width=\linewidth]{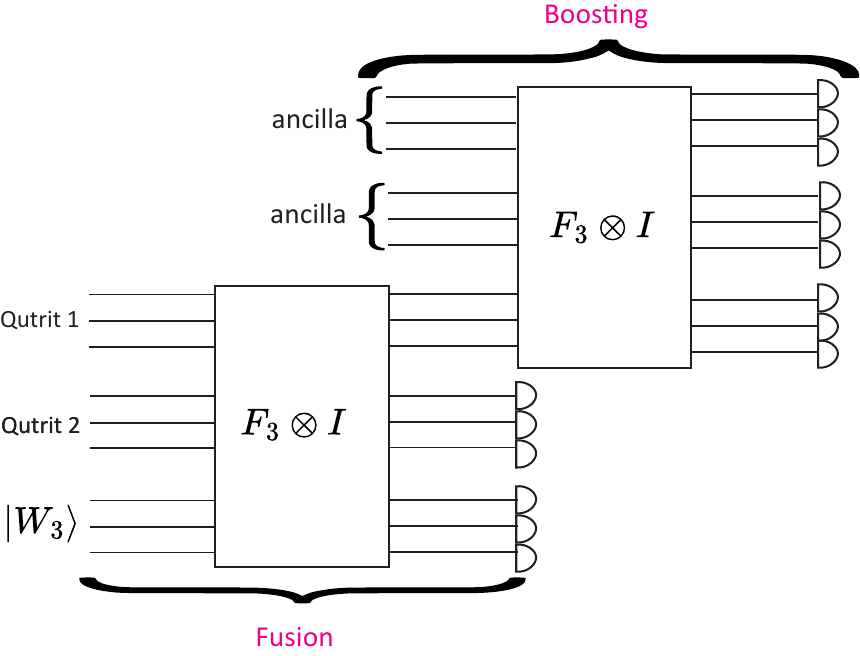}
    \caption{An extension of the qubit boosting protocol of \cite{Ewert_2014} to the qutrit ($d=3$) case. 
    The ``fusion'' part of the circuit follows \citet{Luo} (see Section~\ref{sec:luo original}) and uses a $W$ state ancilla. 
    One output port is then fed into the ``boosting'' circuit, undergoing another Fourier transform along with new ancillary states $\frac{1}{\sqrt{3}}\left(\ket{300}+\ket{030}+\ket{003}\right)$. 
    }
    \label{fig:qutrit_boosting}
\end{figure}

The circuit is depicted for $d=3$ in Fig.~\ref{fig:qutrit_boosting}. 
We start with the circuit of Fig.~\ref{fig:luo-fusion}, then, on any number of output ports (for ease of illustration, only the $0$th), we allow for additional interference. 
The circuit of Fig.~\ref{fig:boosted_type-II} \cite{Ewert_2014} applies Hadamard transformations\textemdash two-body Fourier transforms\textemdash between corresponding modes of an output port and an ancillary state of $2$ photons in $2$ modes, $\frac{1}{\sqrt{2}}\left(\ket{20}+\ket{02}\right)$. 
We straightforwardly extend this, using an interferometer of the form $F_d\otimes I$ to allow for interference between the chosen output port and $d-1$ ancillary states of $d$ photons in $d$ modes, namely 

\begin{equation}\label{eq:appendix ancilla}
    \frac{1}{\sqrt{d}}\left(\ket{d00\cdots 0}+\ket{0d0\cdots 0}+\dots \ket{000\cdots d}\right).
\end{equation}
The key feature of this ancilla is that, for every term, the number of photons in each mode is a multiple of $d$. 
This ensures that the success probability will not be decreased: extending the notion of ``time bin'' to the modes of the ancillary states as well, the total number of photons in each of the $d$ time bins (modulo $d$) remains unchanged. 
Then the ``success'' patterns in the protocols of Section~\ref{sec:luo} remain distinguishable from one another (and from failure patterns) in this setting. 
Our goal is to convert some failure patterns into successes. 
As in \cite{Ewert_2014}, this will only be possible for terms in which all $d$ of the original photons exit through the same port of the initial Fourier interferometer $F_d\otimes I$. 
We specifically target ``bunched'' patterns in which all $d$ original photons are sent to the same mode of the $0$th port. The protocol is given here: 

\begin{protocol}\label{proto:appendix}
    Let $d\geq 2$ be an integer. We give a protocol for boosted Type-II fusion of $d$-dimensional qudits as follows. We boost only the $0$th port, but a similar procedure may be applied to other ports. 
    \begin{enumerate}
        \item Begin with two arbitrary ``input'' qudits (those to be fused) and $d-2$ ancillary W-state qudits. 
        \item Input these $d$ qudits into a $F_d\otimes I$ interferometer, as in the Fourier projection (Protocol~\ref{proto:fourier}). 
        \item For the first $d$ modes, making up the $0$th port, feed their output into a \emph{second} Fourier interferometer $F_d\otimes I$, along with $d-1$ ancillary states of the form \eqref{eq:appendix ancilla}. 
        \item Perform PNRD on all modes (of which there are $d^2 + (d^2-d) = d(2d-1)$). Let $n_i$ be the number of photons in modes indexed by $i$ mod $d$ (in other words, $n_i$ is the number of photons in ``time bin'' $i$). 
        \item Interpret the measurement results: 
        \begin{enumerate}
            \item If all $n_i\equiv 1\mod d$, this uniquely determines an ``unboosted'' success pattern; we then proceed as in Protocol~\ref{proto:fourier}. 
            \item Suppose that the final $d(d-1)$ modes are empty (so that all photons are in the $d^2$ ``boosting'' modes), and all $n_i=d$ (so that each time bin receives the same number of photons). 
            For $0\leq j\leq d-1$, let $m_j$ be the number of photons in boosting modes $jd, jd+1, \dots, jd+d-1$. 
            If we further have $\sum_{j=0}^{d-1} j m_j \equiv 0\mod d$, this leads to a ``boosted'' success. 
            The resulting projection is onto a maximally entangled two-qudit state with no need for extra-dimensional corrections. 
            \item If neither of the previous parts imply success, then the protocol has failed. 
        \end{enumerate}
    \end{enumerate}
\end{protocol}

We briefly discuss numerics for the $d=3$ case before giving intuition for the broader scheme. 
We may perform this boosting operation on any subset of ports; 
for each relevant port, we use $6$ new ancillary photons (now in nontrivial entangled states) to increase the overall success probability by approximately $1.83\times 10^{-3}$. 
Carrying out this procedure on all $3$ ports requires $18$ ancillary photons (in addition to the usual $W$ state ancilla) and provides a total success probability of $0.117$, a small increase over the ``unboosted'' success probability $1/9\approx 0.111$. 
This is still significantly worse than the probability of $1/6\approx 0.167$ we obtain by using Protocol~\ref{proto:odd}, which requires only $2$ ancillary photons. 

We now return to the general case, starting with the original circuit marked ``Fusion'' in Fig.~\ref{fig:qutrit_boosting}. 
Note that, if we apply $F_d\otimes I$ and post-select for terms with all photons in the $0$th port, we obtain
\begin{align}
    \ketb{k,k,\dots, k}\mapsto \ket{d_k}\ket{0}^{\otimes d(d-1)}
\end{align}
with probability 
\begin{equation}
    c = \bra{d,0,\dots, 0}F_d \ket{1, 1, \dots, 1} = \sqrt{\dfrac{d!}{d^d}},
\end{equation}
where $\ket{d_k}$ is the $d$-mode bunched state with $d$ photons in the $k$th mode. 
In particular, under the same operation, we have 
\begin{equation}\label{eq:appendix bunched}
    \frac{1}{\sqrt{d}}\sum_{k=0}^{d-1}\ketb{k}^{\otimes d}\mapsto \left(\frac{1}{\sqrt{d}}\sum_{k=0}^{d-1} \ket{d_k}\right)\ket{0}^{\otimes d(d-1)}
    \end{equation}
with the same probability $c$. 
Then we can obtain a projection onto the $d$-qudit GHZ state if we can design the boosting circuit to project onto the uniform superposition of the $\ket{d_k}$. 
This is done by using the ancillary states \eqref{eq:appendix ancilla}, applying another Fourier transform $F_d\otimes I$, and post-selecting according to the rules in Protocol~\ref{proto:appendix}. 
The condition that all photons go to the $0$th port makes the above calculation relevant. 
The condition that all $n_i=d$ forces a projection onto some superposition of the $\ket{d_k}$ (and intuitively encourages uniform superpositions over non-uniform ones). 
Finally, the condition $\sum_{j=0}^{d-1} j m_j \equiv 0\mod d$ is related to the zero-transmission law \cite{tichy2010zero} and uses the symmetries of the ideal input state \cite{dittel2018totally, saied2025general} to obtain the desired projection. 
We conjecture that the probability of the state \eqref{eq:appendix bunched} is precisely the scalar $c^2$ above. 
Putting this all together, and accounting for a factor of $1/d^{d}$ for projecting the GHZ state to a random two-qudit state with $W$-state ancillae, the \emph{increase} in success probability given by Protocol~\ref{proto:appendix} is then conjectured to be
\begin{equation}
    d!^2/d^{3d}.
\end{equation}
We have verified this numerically in the case $d=3$ and expect that the general case is similar. 
This requires $O(d^2)$ additional photons to increase the success probability by an amount that exponentially decays with $d$. 
Thus, we expect Protocol~\ref{proto:appendix} is not practically applicable, especially in the presence of errors such as photon loss. 
However, we discuss it here since it is a natural generalization of the works of \cite{Luo, Ewert_2014} that others may have been interested in investigating. 
Further, the appearance of the zero-transmission law \cite{tichy2010zero} in the success patterns is theoretically interesting and may inspire more efficient Type-II fusion protocols in the future. 

\bibliography{refs} 
\end{document}